\documentclass[11pt,a4paper]{article}

\usepackage{amssymb,amsthm,amsmath,mathtools}
\usepackage{thmtools,thm-restate}
\usepackage[round]{natbib}
\usepackage{pifont}

\usepackage[usenames,svgnames,xcdraw,table]{xcolor}
\definecolor{DarkGreen}{rgb}{0.1,0.5,0.1}
\usepackage[backref=page]{hyperref}
\hypersetup{
	colorlinks=true,
	linkcolor=red,
	urlcolor=DarkGreen,
	citecolor=blue
}
\renewcommand*{\backref}[1]{}
\renewcommand*{\backrefalt}[4]{%
    \ifcase #1 (Not cited.)%
    \or        (Cited on page~#2)%
    \else      (Cited on pages~#2)%
    \fi}
\usepackage[capitalise,noabbrev]{cleveref}
\Crefname{property}{Property}{Properties}
\Crefname{thm}{Theorem}{Theorems}
\Crefname{example}{Example}{Examples}
\Crefname{table}{Table}{Tables}

\usepackage{cancel}
\usepackage{tabularx}
\usepackage{nicefrac}
\usepackage{enumerate}
\usepackage{etoolbox}
\usepackage[margin=1in]{geometry}
\usepackage[T1]{fontenc}
\usepackage[tt=false]{libertine}

\usepackage{url}
\usepackage[utf8]{inputenc}
\usepackage[small]{caption}
\usepackage{graphicx}
\usepackage{booktabs}
\usepackage{mathrsfs,amsfonts,dsfont,authblk}
\usepackage{array,multirow,graphicx,bigdelim}

\usepackage[linesnumbered,lined,boxed,ruled,vlined]{algorithm2e}

\SetKwInOut{Parameters}{Parameters}
\SetKwComment{Comment}{$\triangleright$\ }{}
\SetAlFnt{\small}
\SetAlCapFnt{\small}
\SetAlCapNameFnt{\small}
\SetAlCapHSkip{0pt}
\IncMargin{-\parindent}

\usepackage{tikz}
\usepackage{tikz-cd}
\usetikzlibrary{fit,calc,shapes,decorations.text,arrows,decorations.markings,decorations.pathmorphing,shapes.geometric,positioning,decorations.pathreplacing}
\tikzset{snake it/.style={decorate, decoration=snake}}

\newcommand*{\tikzmk}[1]{\tikz[remember picture,overlay,] \node (#1) {};\ignorespaces}
\newcommand{\boxit}[1]{\tikz[remember picture,overlay]{\node[yshift=3pt,xshift=4pt,fill=#1,opacity=.25,fit={(A)($(B)+(1.0\linewidth,.8\baselineskip)$)}] {};}\ignorespaces}
\colorlet{mygray}{gray!40}

\makeatletter
\let\oldnl\nl
\newcommand{\nonl}{\renewcommand{\nl}{\let\nl\oldnl}}
\makeatother

  {\list{}{\leftmargin=#1\rightmargin=#1}\item[]}%
  {\endlist}

\usepackage[labelfont={normalfont,bf},textfont=it]{caption}
\usepackage{subcaption}

\theoremstyle{definition}

\newtheorem{dfn}{Definition}

\newtheorem{examplex}{Example}
\newenvironment{example}{\pushQED{\qed}\examplex}{\popQED\endexamplex}
\theoremstyle{remark}

\usepackage{flushend}
\usepackage[utf8]{inputenc}
\usepackage[inline]{enumitem}
\Crefname{claim}{Claim}{Claims}

\allowdisplaybreaks

\newcommand{\stp}{\text{strategyproofness}}

\renewcommand{\>}{{\succ}}

\newcommand{\MMS}{\textrm{\textup{MMS}}}

\newcommand{\NP}{\textrm{\textup{NP}}}
\newcommand{\NPH}{\textrm{\textup{NP-hard}}}
\newcommand{\NPC}{\textrm{\textup{NP-complete}}}
\newcommand{\EF}[1]{\ifstrempty{#1}{\textrm{\textup{EF}}}{\textrm{\textup{EF{$#1$}}}}}
\newcommand{\EFX}{\textrm{\textup{EFX}}}
\renewcommand{\L}{\mathcal{L}}
\newcommand{\RM}{\textrm{\textup{RM}}}

\newcommand{\PITFull}{\textup{\textsc{Partition Into Triangles}}}
\newcommand{\PIT}{\textup{\textsc{PIT}}}
\newcommand{\PO}{\textup{PO}}

\newcommand{\SDQ}{\textrm{\textup{SDQ}}}

\newcommand{\SP}{\textrm{\textup{SP}}}

\newcommand{\ThreeSAT}{\textup{\textsc{3-SAT}}}
\newcommand{\TTTSAT}{\textup{\textsc{(2/2/3)-SAT}}}

\newcommand{\EFkRMExistence}[1]{\ifstrempty{#1}{\textup{\textsc{\EF{k}+\RM{}-Existence}}}{\textup{\textsc{\EF{#1}+\RM{}-Existence}}}}

\begin{document}

\title{Fair and Efficient Allocations under Lexicographic Preferences}
\date{}

\author[1]{Hadi Hosseini}
\author[2]{Sujoy Sikdar}
\author[3]{Rohit Vaish}
\author[4]{Lirong Xia}
\affil[1]{Pennsylvania State University\\
	{\small\texttt{hadi@psu.edu}}}
\affil[2]{Binghamton University\\
	{\small\texttt{ssikdar@binghamton.edu}}}
\affil[3]{Tata Institute of Fundamental Research\\
	{\small\texttt{rohit.vaish@tifr.res.in}}}
\affil[4]{Rensselaer Polytechnic Institute\\
	{\small\texttt{xial@cs.rpi.edu}}}

\maketitle

\begin{abstract}
Envy-freeness up to any good (\EFX{}) provides a strong and intuitive guarantee of fairness in the allocation of indivisible goods. But whether such allocations always exist or whether they can be efficiently computed remains an important open question. We study the existence and computation of \EFX{} in conjunction with various other economic properties under \emph{lexicographic preferences}--a well-studied preference model in artificial intelligence and economics. In sharp contrast to the known results for additive valuations, we not only prove the existence of \EFX{} and Pareto optimal allocations, but in fact provide an algorithmic characterization of these two properties. We also characterize the mechanisms that are, in addition, strategyproof, non-bossy, and neutral. When the efficiency notion is strengthened to rank-maximality, we obtain non-existence and computational hardness results, and show that tractability can be restored when \EFX{} is relaxed to another well-studied fairness notion called maximin share guarantee (\MMS{}).
\end{abstract}

\section{Introduction}

Fair and efficient allocation of scarce resources is a fundamental problem in economics and computer science. The quintessential fairness notion---\emph{envy-freeness}---enjoys strong existential and computational guarantees for \emph{divisible} resources~\citep{varian1974equity}. However, in notable applications such as course allocation~\citep{budish2011combinatorial} and property division~\citep{PW12divorcing} that involve \emph{indivisible} resources, (exact) envy-freeness could be too restrictive. In these settings, it is natural to consider notions of approximate fairness such as \emph{envy-freeness up to any good} (\EFX{}) wherein pairwise envy can be eliminated by the removal of any good in the envied bundle~\citep{caragiannis2019unreasonable}.

\EFX{} is arguably the closest analog of envy-freeness in the indivisible setting, and, as a result, has been actively studied especially for the domain of additive valuations. However, it also suffers from a number of limitations: First, barring a few special cases, the existence and computation of \EFX{} allocations remains an open problem. Second, for additive valuations, \EFX{} can be incompatible with \emph{Pareto optimality} (\PO{})---a fundamental notion of economic efficiency~\citep{PR20almost}. Finally, \EFX{} could also be at odds with \emph{strategyproofness}~\citep{amanatidis2017truthful}, which is another desirable property in the economic analysis of allocation problems.

The aforementioned limitations of \EFX{} prompt us to explore the \emph{domain restriction} approach in search of positive results~\citep{ELP16preference}. Specifically, we deviate from the framework of cardinal preferences for which \EFX{} allocations have been most extensively studied, and instead focus on the purely ordinal domain of \emph{lexicographic preferences}.  

Lexicographic preferences have been widely studied in psychology~\citep{GG96reasoning}, machine learning~\citep{SM06complexity}, and social choice~\citep{T70problem} as a model of human decision-making. Several real-world settings such as evaluating job candidates and the desirability of a product involve lexicographic preferences over the set of features. In the context of fair division, too, lexicographic preferences can arise naturally. For example, when dividing an inheritance consisting of a house, a car, and some home appliances, a stakeholder might prefer any division in which she gets the house over one where she doesn't (possibly because of its sentimental value), subject to which she might prefer any outcome that includes the car over one that doesn't, and so on.

On the computational side, lexicographic preferences provide a succinct language for representing preferences over combinatorial domains~\citep{saban2014note,LMX18voting}, and have led to numerous positive results at the intersection of artificial intelligence and economics~\citep{FLS18complexity,hosseini2019multiple}. Motivated by these considerations, our work examines the existence and computation of \emph{fair} (i.e., \EFX{}) and \emph{efficient} allocations from the lens of lexicographic preferences.

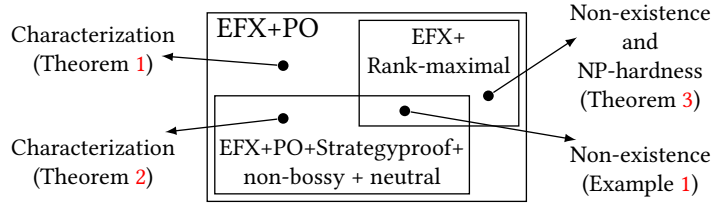
\begin{figure}[h]
    \centering
	\tikzset{every picture/.style={line width=0.5pt}}
	\begin{tikzpicture}
		\footnotesize
		\def\toprightX{4.2}
		\def\toprightY{2.5}
		\draw (0,0) rectangle (\toprightX,\toprightY);
		\node (1) at (0.8,\toprightY-0.2) {\normalsize{\EFX{}+\PO{}}};
		\draw (\toprightX/2-0.1,\toprightY/2-0.25) rectangle (\toprightX-0.1,\toprightY-0.1);
		\node[align=center] (2) at (3*\toprightX/4-0.1,\toprightY-0.5) {{\EFX{}+}\\{Rank-maximal}};
		\draw (0.1,0.1) rectangle (3*\toprightX/4+0.3,\toprightY-1.1);
		\node[align=center] (3) at (3*\toprightX/8+0.2,\toprightY-2) {{\EFX{}+\PO{}+Strategyproof+}\\{non-bossy + neutral}};
		\draw (-1.5,\toprightY-0.5) node[align=center] (8) {{Characterization}\\{(\Cref{thm:EFX_PO})}};
		\draw (\toprightX+1.5,\toprightY-0.6) node[align=center] (9) {{Non-existence}\\{and}\\{\NPH{}ness}\\{(\Cref{thm:EFX_RM_NP-complete})}};
		\draw (\toprightX+1.5,0.4) node[align=center] (10) {{Non-existence}\\{(\Cref{eg:SP+RM_NonExistence})}};
		\draw (-1.5,0.5) node[align=center] (11) {{Characterization}\\{(\Cref{thm:EFX_PO_SP_Neutral_NonBossy})}};
		\node[circle,fill=black,minimum size=4pt,inner sep=0pt] (4) at (1,\toprightY-0.7) {};
		\node[circle,fill=black,minimum size=4pt,inner sep=0pt] (5) at (1,1.1) {};
		\node[circle,fill=black,minimum size=4pt,inner sep=0pt] (6) at (\toprightX/2+0.5,\toprightY-1.3) {};
		\node[circle,fill=black,minimum size=4pt,inner sep=0pt] (7) at (\toprightX-0.5,\toprightY-1.1) {};
		\draw[shorten >=0.9cm,-latex] (4) -- (8.center);
		\draw[shorten >=-0.25cm,-latex] (5) -- (11.north east);
		\draw[shorten >=-0.2cm,-latex] (6) -- (10.north west);
		\draw[shorten >=1cm,-latex] (7) -- (9.north);
	\end{tikzpicture}    
    \caption{Summary of our theoretical results.}
    \label{fig:Summary_Of_Results}
\end{figure}

\paragraph{Our Contributions.}
\Cref{fig:Summary_Of_Results} summarizes our theoretical contributions.

\begin{itemize}
	\item \textbf{\EFX{}+\PO{}}: Our first result provides a family of polynomial-time algorithms for computing \EFX{}+\PO{} allocations under lexicographic preferences. Furthermore, we show that \emph{any} \EFX{}+\PO{} allocation can be computed by some algorithm in this family, thus providing an algorithmic characterization of such allocations~(\Cref{thm:EFX_PO}). This result establishes a sharp contrast with the additive valuations domain where the two properties are incompatible in general.

	\item \textbf{\EFX{}+\PO{}+strategyproofness}:  The positive result for \EFX{}+\PO{} motivates us to investigate a more demanding property combination of \EFX{}, \PO{}, and strategyproofness. Once again, we obtain an algorithmic characterization~(\Cref{thm:EFX_PO_SP_Neutral_NonBossy}): Subject to some common axioms (non-bossiness and neutrality), any mechanism satisfying \EFX{}, \PO{}, and strategyproofness is characterized by a special class of \emph{quota-based serial dictatorship mechanisms}~\citep{Papai00:Strategyproofquotas,hosseini2019multiple}.

	\item \textbf{\EFX{}+rank-maximality}: When the efficiency notion is strengthened to \emph{rank-maximality}, we encounter incompatibility with strategyproofness (\Cref{eg:SP+RM_NonExistence}) as well as with \EFX{} (\Cref{eg:EFk+RM_NonExistence}). Furthermore, checking the existence of \EFX{} and rank-maximal allocations turns out to be \NPC{}~(\Cref{thm:EFX_RM_NP-complete}), suggesting that our algorithmic results are, in a certain sense, `maximal'. The intractability persists even when \EFX{} is relaxed to envy-freeness up to $k$ goods ($\EF{k}$)~(\Cref{thm:EFk_RM_NP-complete}), but efficient computation is possible if \EFX{} is relaxed to another well-studied fairness notion called \emph{maximin share guarantee} or \MMS{}~(\Cref{thm:MMS_RM_Goods}).
\end{itemize}

\paragraph{Related Work.}
Envy-free solutions may not always exist for indivisible goods. As a result, the literature has focused on notions of approximate fairness, most notably \emph{envy-freeness up to one good} (\EF{1}) and its strengthening called \emph{envy-freeness up to any good} (\EFX{}). The former enjoys strong existential and algorithmic support, as an \EF{1} allocation always exists for general monotone valuations and can be efficiently computed. However, achieving \EF{1} together with economic efficiency seems non-trivial: For additive valuations, \EF{1}+\PO{} allocations always exist~\citep{caragiannis2019unreasonable,BKV18finding} but no polynomial-time algorithm is known for computing such allocations.

The stronger notion of \EFX{} has proven to be more challenging. As mentioned previously, the existence of \EFX{} for additive valuations remains an open problem. Additionally, \EFX{} and Pareto optimality are known to be incompatible for non-negative additive valuations~\citep{PR20almost} and is open for positive additive valuations.

The aforementioned limitations of \EFX{} have motivated the study of further relaxations or special cases in search of positive results. Some recent results establish the existence of partial allocations that satisfy \EFX{} after discarding a small number of goods while also fulfilling certain efficiency criteria~\citep{CGH19envy,CKM+20little}. Similarly, \EFX{} allocations have been shown to exist for the special case of three agents with additive valuations~\citep{CGM20efx}, or when the agents can be partitioned into two \emph{types}~\citep{M20existence}, or when agents have dichotomous preferences~\citep{ABF+20maximum}. For cardinal utilities, various multiplicative approximations of \EFX{} (and its variant that involves removing an \emph{average} good) have been considered~\citep{PR20almost,AMN20multiple,CGM20fair,FHL+20almost}. Another emerging line of work studies \EFX{} for \emph{non-monotone} valuations, i.e., when the resources consist of both goods and chores~\citep{CL20fairness,BBB+20envy}.

The interaction between fairness and efficiency is further complicated with the addition of \emph{strategyproofness} due to several fundamental impossibility results both in deterministic \citep{zhou1990conjecture} as well as randomized settings~\citep{Bogomolnaia01:New,kojima2009random}.

Indeed, while ordinal efficiency is compatible with envy-freeness, such outcomes cannot, in general, be achieved via (weakly) strategyproof mechanisms even under strict preferences~\citep{kojima2009random}. Moreover, sd-efficiency and sd-strategyproofness (here, \emph{sd} stands for stochastic dominance) are incompatible even with a weak notion of stochastic fairness called equal treatment of equals~\citep{aziz2017impossibilities}. In a similar vein, for deterministic mechanisms, any strategyproof mechanism could fail to satisfy \EF{1} even for two agents under additive valuations~\citep{amanatidis2017truthful}.

Lexicographic preferences have been successfully used as a domain restriction to circumvent impossibility results in mechanism design~\citep{Sikdar2017:Mechanism,FLS18complexity}. In fair division of indivisible goods, lexicographic (sub)additive utilities have facilitated constant-factor approximation algorithms for egalitarian and Nash social welfare objectives~\citep{BBL+17positional,N20fairly}. \citet{hosseini2019multiple} show that under lexicographic preferences, a mechanism is Pareto optimal, strategyproof, non-bossy, and neutral if and only if it is a serial dictatorship quota mechanism. In randomized settings, too, lexicographic preferences have led to the design of mechanisms that simultaneously satisfy stochastic efficiency, envy-freeness, and \stp{}~\citep{SV15allocation,hosseini2019multiple}.

\section{Preliminaries}

\paragraph{Model} For any $k \in \mathbb{N}$, define $[k] \coloneqq \{1,\dots,k\}$. An \emph{instance} of the allocation problem is a tuple $\langle N, M, \succ \rangle$, where $N \coloneqq [n]$ is a set of $n$ {\em agents}, $M$ is a set of $m$ \emph{goods}, and $\succ \, \coloneqq (\succ_1, \dots, \succ_n)$ is a {\em preference profile} that specifies the ordinal preference of each agent $i \in N$ as a linear order $\succ_i \, \in \L$ over the set of goods; here, $\L$ denotes the set of all (strict and complete) linear orders over $M$.

\paragraph{Allocation and bundles} A \emph{bundle} is any subset $X \subseteq M$ of the set of goods. An {\em allocation} $A=(A_1,\dots,A_n)$ is an $n$-partition of $M$, where $A_i\subseteq M$ is the bundle assigned to agent $i$. We will write $\Pi$ to denote the set of all $n$-partitions of $M$. We say that allocation $A$ is {\em partial} if $\bigcup_{i \in N} A_i \subset M$, and \emph{complete} if $\bigcup_{i \in N} A_i = M$.

\paragraph{Lexicographic preferences}
We will assume that agents' preferences over the bundles are given by the lexicographic extension of their preferences over individual goods. Informally, this means that if an agent ranks the goods in the order $a \succ b \succ c \succ \dots$, then it prefers a bundle containing $a$ over any other bundle that doesn't, subject to that, it prefers a bundle containing $b$ over any other bundle that doesn't, and so on. Formally, given any pair of bundles $X,Y \subseteq M$ and any linear order $\succ_i \, \in \L$, we have $X \succ_i Y$ if and only if there exists a good $g \in X \setminus Y$ such that $\{g' \in Y : g' \succ_i g\} \subseteq X$. Notice that since $\succ_i$ is a linear order over $M$, the corresponding lexicographic extension is a linear order over $2^M$.

For any agent $i \in N$ and any pair of bundles $X,Y \in M$, we will write $X \succeq_i Y$ if either $X \succ_i Y$ or $X=Y$.

\paragraph{Envy-freeness} Given a preference profile $\>$, an allocation $A$ is said to be (a) \emph{envy-free} (\EF{}) if for every pair of agents $i,h \in N$, we have $A_i \succeq_i A_h$; (b) \emph{envy-free up to any good} (\EFX{}) if for every pair of agents $i,h \in N$ such that $A_h \neq \emptyset$ and every good $j \in A_h$, we have $A_i \succeq_i A_h \setminus \{j\}$, and (c) \emph{envy-free up to $k$ goods} (\EF{k}) if for every pair of agents $i,h \in N$ such that $A_h \neq \emptyset$, there exists a set $S \subseteq A_h$ such that $|S| \leq k$ and $A_i \succeq_i A_h \setminus S$. Clearly, $\EFX{} \Rightarrow \EF{1} \Rightarrow \EF{2} \Rightarrow  \dots$.

\paragraph{Maximin Share} An agent's maximin share is its most preferred bundle that it can guarantee itself as a divider in an $n$-person cut-and-choose procedure against adversarial opponents~\citep{budish2011combinatorial}. Formally, the maximin share of agent $i$ is given by $\MMS_i \coloneqq \max_{A \in \Pi} \min_{i} \{A_1,\dots,A_n\}$, where $\min\{\cdot\}$ and $\max\{\cdot\}$ denote the least-preferred and most-preferred bundles with respect to $\succ_i$. An allocation $A$ satisfies \emph{maximin share guarantee} (\MMS{}) if each agent receives a bundle that it weakly prefers to its maximin share. That is, the allocation $A$ is \MMS{} if for every $i \in N$, $A_i \succeq_i \MMS_i$. It is easy to see that $\EF{} \Rightarrow \MMS{}$. Additionally, for lexicographic preferences, we have that $\EFX{} \Rightarrow \MMS{}$ (the converse is not true) while \EF{1} and \MMS{} can be incomparable (see \Cref{sec:app:prop:EFX_implies_MMS} of the appendix).

\paragraph{Pareto optimality} Given a preference profile $\>$, an allocation $A$ is said to be \emph{Pareto optimal} (\PO{}) if there is no other allocation $B$ such that $B_i \succeq_i A_i$ for every agent $i \in N$ and $B_k \succ_k A_k$ for some agent $k \in N$. 

\paragraph{Rank-maximality} A \emph{rank-maximal} (\RM{}) allocation is one that maximizes the number of agents who receive their favorite good, subject to which it maximizes the number of agents who receive their second favorite good, and so on~\citep{IKM+06rank,P13capacitated}. Given an allocation $A$, its \emph{signature} refers to a tuple $(n_1,n_2,\dots,n_m)$ where $n_i$ is the number of agents who receive their $i^\text{th}$ favorite good (note that an agent can contribute to multiple $n_i$'s). All rank-maximal allocations for a given instance have the same signature. Computing \emph{some} rank-maximal allocation for a given instance is easy: Assign each good to an agent that ranks it the highest among all agents (tiebreak arbitrarily). This procedure provides a computationally efficient way of computing the signature of a rank-maximal allocation as well as verifying whether a given allocation is rank-maximal. Notice that rank-maximality is a strictly stronger requirement than Pareto optimality.

\paragraph{Mechanism} A mechanism $f: \L^n \to \Pi$ is a mapping from preference profiles to allocations. For any preference profile $\succ \, \in \L^n$, we use $f(\>)$ to denote the allocation returned by $f$, and $f_i(\>)$ to denote the bundle assigned to agent $i$.

\paragraph{Properties of mechanisms} A mechanism $f: \L^n \to \Pi$ is said to satisfy \EF{}\,/\,\EFX{}\,/\,\EF{k}\,/\,\PO{}\,/\,\RM{} if for every preference profile $\succ \, \in \L^n$, the allocation $f(\>)$ has that property. In addition, a mechanism $f$ satisfies
\begin{itemize}
	\item \emph{strategyproofness} (\SP{}) if no agent can improve by misreporting its preferences. That is, for every preference profile $\succ \, \in \L^n$, every agent $i \in N$, and every (misreported) linear order $\>'_i \in \L$, we have $f_i(\>) \succeq_i f_i(\>')$, where $\>' \coloneqq (\>_1,\dots,\>_{i-1},\>'_i,\>_{i+1},\dots,\>_n)$.
	
	\item \emph{non-bossiness} if no agent can modify the allocation of another agent by misreporting its preferences without changing its own allocation. That is, for every profile $\succ \, \in \L^n$, every agent $i\in N$, and every (misreported) linear order $\>'_i \in \L$, we have  $f_i(\>')=f_i(\>) \Rightarrow f(\>')=f(\>)$, where $\>' \coloneqq (\>_1,\dots,\>_{i-1},\>'_i,\>_{i+1},\dots,\>_n)$.
	
	\item \emph{neutrality} if relabeling the goods results in a consistent change in the allocation. That is, for every preference profile $\succ \, \in \L^n$ and every relabeling of the goods $\pi: M \rightarrow M$, it holds that $f(\pi(\>)) = \pi(f(\>))$, where $\pi(\>) \coloneqq (\pi(\>_1),\dots,\pi(\>_n))$ and $\pi(A) \coloneqq (\pi(A_1),\dots,\pi(A_n))$ for any allocation $A = (A_1,\dots,A_n)$.
\end{itemize}

\section{\EFX{} and Pareto Optimality}

Recall that for additive valuations, establishing the existence of \EFX{} allocations remains an open problem, and there exist instances where no allocation is simultaneously \EFX{} and \PO{}~\citep{PR20almost}. Our first result (\Cref{thm:EFX_PO}) shows that there is no conflict between fairness and efficiency for lexicographic preferences: Not only does there exist a family of polynomial-time algorithms that always return \EFX{}+\PO{} allocations, but \emph{every} \EFX{}+\PO{} allocation can be computed by some algorithm in this family. We will start with an easy observation concerning \EFX{} allocations.

\begin{restatable}{prop}{EFXProperty}
An allocation $A$ is \EFX{} if and only if each envied agent in $A$ gets exactly one good.
\label{prop:efx_property}
\end{restatable}

\paragraph{Description of algorithm}
Each algorithm in the family (Algorithm~\ref{alg:EFX+PO}) is specified by an ordering $\sigma$ over the agents, and consists of two phases. Phase 1 involves a single round of serial dictatorship according to $\sigma$. Phase 2 assigns the remaining goods among the \emph{unenvied agents} according to a picking sequence $\tau$. Note that the set of unenvied agents after Phase 1 must be nonempty; in particular, the last agent in $\sigma$ belongs to this set since every other agent prefers the good that it picked in Phase 1 over any good in the last agent's bundle.

\begin{algorithm}[h]
\DontPrintSemicolon
 \linespread{1.2}
\KwIn{An instance $\langle N, M, \> \rangle$ with lexicographic preferences}
\Parameters{A permutation $\sigma: N \rightarrow N$ of the agents}
\KwOut{An allocation $A$}
$A \leftarrow (\emptyset,\dots,\emptyset)$\;
\Comment{\scriptsize{Phase 1: Serial dictatorship for assigning $n$ goods}}
\tikzmk{A}
Agents arrive according to $\sigma$, and each picks a favorite good from the set of remaining goods. Update the partial allocation $A$.\;
\nonl \tikzmk{B}
\boxit{mygray}
\Comment{\scriptsize{Phase 2: Allocate leftover goods via picking sequence}}
\oldnl \tikzmk{A}
\uIf{the set of remaining goods is nonempty}{$U \leftarrow \{i \in N : i \text{ is not envied by any agent under $A$}\}$.\;
Fix any picking sequence $\tau$ of length $m-n$ consisting only of the agents in $U$ (i.e., the \emph{unenvied} agents).\;
Assign remaining goods according to $\tau$ and update $A$.\;}
\tikzmk{B}
\boxit{mygray}
\KwRet{$A$}
\caption{\EFX{}+\PO{}}
\label{alg:EFX+PO}
\end{algorithm}

\begin{restatable}{thm}{thmefxpo}
For any ordering $\sigma$ of the agents, the allocation computed by Algorithm~\ref{alg:EFX+PO} satisfies \EFX{} and \PO{}. Conversely, any \EFX{}+\PO{} allocation can be computed by Algorithm~\ref{alg:EFX+PO} for some choice of $\sigma$.
\label{thm:EFX_PO}
\end{restatable}
\begin{proof}
We will start by showing that the allocation $A$ returned by Algorithm~\ref{alg:EFX+PO} satisfies \EFX{}. From \Cref{prop:efx_property}, it suffices to show that any envied agent gets exactly one good in $A$. Notice that any agent that is envied at the end of Phase 1 does not receive any good in Phase 2. Furthermore, the pairwise envy relations remain unchanged during Phase 2 since each agent has already picked its favorite available good in Phase 1, and because of lexicographic preferences, any goods assigned in Phase 2 are strictly less preferred. Thus, $A$ is \EFX{}.

To prove Pareto optimality (\PO{}), suppose, for contradiction, that $A$ is Pareto dominated by an allocation $B$. Then, there must exist some agent, say $i$, who receives a good under $A$ that it does not receive under $B$ (i.e., $A_i \setminus B_i \neq \emptyset$); we will call any such item a \emph{difference} good. Observe that the execution of Algorithm~\ref{alg:EFX+PO} can be described in terms of a combined picking sequence $\langle \sigma, \tau \rangle$. Thus, without loss of generality, we can define $i$ to be the \emph{first} agent according to $\langle \sigma, \tau \rangle$ to receive a difference good. Let $g$ denote the corresponding difference good picked by $i$, and note that $g \in A_i \setminus B_i$ by assumption.

Since $B$ Pareto dominates $A$ and $A_i \neq B_i$, we must have $B_i \, \>_i \, A_i$. For lexicographic preferences, this means that there exists a good $g' \in B_i \setminus A_i$ such that $g' \, \>_i \, g$. Since all agents preceding $i$ in $\langle \sigma, \tau \rangle$ pick the goods that they also own under $B$, the good $g'$ must be available (along with $g$) when it is $i$'s turn to pick. Thus, $i$ would not pick $g$, which is a contradiction. Hence, $A$ satisfies \PO{}.

To prove the converse, note that any \PO{} allocation can be induced by a picking sequence.\footnote{Indeed, in any \PO{} allocation, some agent must receive its favorite good (otherwise a cyclic exchange of the top-ranked goods gives a Pareto improvement). Add this agent to the picking sequence, and repeat the procedure for the remaining goods.} Given any \EFX{}+\PO{} allocation $A$, let $S$ denote the corresponding picking sequence. We claim that without loss of generality, the first $n$ positions in $S$ belong to $n$ different agents. Indeed, if some agent $i$ appears more than once in the $n$-prefix of $S$, then $|A_i|>1$. By \Cref{prop:efx_property}, $i$ must not be envied by any other agent. For lexicographic preferences, this means that the good picked by any other agent $j$ in its first appearance in $S$ is preferred by $j$ over every good picked by $i$.

Without loss of generality, let $i$ be the first agent with a repeated occurrence in $S$. Let $t_i$ denote the index (i.e., position in $S$) of the second appearance of $i$. Among all agents whose first appearance occurs after $t_i$, let $j$ denote the first one, and suppose this appearance occurs at position $t_j$ in the sequence $S$. Then, all positions between $t_i$ and until (but excluding) $t_j$ correspond to repeated occurrences. By the aforementioned argument, $j$ does not envy any of the corresponding agents. Furthermore, none of the corresponding agents prefer the good picked by $j$ over the ones that they picked, since they appear before $j$ in the sequence $S$.

Thus, a modified sequence where $j$ is pushed immediately before $t_i$ without making any other changes results in the same allocation. Repeated use of the same observation gives us that the repeated appearances of unenvied agents can be ``pushed behind''' the first appearances of other agents without loss of generality, implying that the $n$-prefix of $S$ is a permutation.

We can now instantiate Algorithm~\ref{alg:EFX+PO} with $\sigma$ as the $n$-prefix of $S$ and $\tau$ as $S \setminus \sigma$ to compute the allocation 
$A$.
\end{proof}

\section{Characterizing Strategyproof Mechanisms}

In addition to fairness and efficiency, an important desideratum for allocation mechanisms is strategyproofness. For additive valuations, strategyproofness is known to be incompatible even with \EF{1}~\citep{amanatidis2017truthful}. By contrast, for lexicographic preferences, we will show that strategyproofness can be achieved in conjunction with a stronger fairness guarantee (\EFX{}) as well as Pareto optimality, non-bossiness, and neutrality~(\Cref{thm:EFX_PO_SP_Neutral_NonBossy}). Indeed, a special case of the mechanism in Algorithm~\ref{alg:EFX+PO} where the last agent gets all the remaining goods characterizes these properties (Algorithm~\ref{alg:IQSD}).

\begin{algorithm}[h]
\DontPrintSemicolon
 \linespread{1.2}
\KwIn{An instance $\langle N, M, \> \rangle$ with lexicographic preferences}
\Parameters{A permutation $\sigma: N \rightarrow N$ of the agents}
\KwOut{An allocation $A$}
$A \leftarrow (\emptyset,\dots,\emptyset)$\;
Execute one round of serial dictatorship according to $\sigma$.\;
Assign all remaining goods to the last agent in $\sigma$.\;
\KwRet{$A$}
\caption{}
\label{alg:IQSD}
\end{algorithm}

Our characterization result builds upon an existing result of \citet[Theorem 5.6]{hosseini2019multiple} (see \Cref{prop:quota}) that characterizes four out of the five properties mentioned above (excluding \EFX{}) in terms of {\em Serial Dictatorship Quota Mechanisms} (\SDQ{}), as defined below.

\begin{dfn}
The Serial Dictatorship Quota $(\SDQ{})$ mechanism is specified by a permutation $\sigma: N \rightarrow N$ of the agents and a set of quotas $(q_1,\dots,q_n)$ such that $\sum_{i=1}^n q_i = m$. Given a lexicographic instance $\langle N, M, \> \rangle$ as input, the \SDQ{} mechanism considers agents in the order $\sigma$, and assigns the $i^\text{th}$ agent its most preferred bundle of size $q_i$ from the remaining goods. The resulting allocation is returned as output.
\label{dfn:sdq}
\end{dfn}

\begin{restatable}[\citealp{hosseini2019multiple}]{prop}{HLquota}
For lexicographic preferences, a mechanism is Pareto optimal, strategyproof, non-bossy, and neutral if and only if it is \SDQ{}.
\label{prop:quota}
\end{restatable}

The next result (\Cref{thm:EFX_PO_SP_Neutral_NonBossy}) provides an algorithmic characterization of \EFX{}, \PO{}, strategyproofness, non-bossiness, and neutrality for lexicographic preferences.

\begin{restatable}{thm}{thmIQSD}
For any ordering $\sigma$ of the agents, Algorithm~\ref{alg:IQSD} is \EFX{}, \PO{}, strategyproof, non-bossy, and neutral. Conversely, any mechanism satisfying these properties can be implemented by Algorithm~\ref{alg:IQSD} for some $\sigma$.
\label{thm:EFX_PO_SP_Neutral_NonBossy}
\end{restatable}

\begin{proof}
Note that Algorithm~\ref{alg:IQSD} is a special case of \SDQ{} for the quotas $q_i = 1$ for all $i \in [n-1]$ and $q_n=m-(n-1)$. Therefore, from \Cref{prop:quota}, it is \PO{}, strategyproof, non-bossy, and neutral. Furthermore, Algorithm~\ref{alg:IQSD} is also a special case of Algorithm~\ref{alg:EFX+PO} and is therefore \EFX{} (\Cref{thm:EFX_PO}).

To prove the converse, let $f$ be an arbitrary mechanism satisfying the desired properties. From \Cref{prop:quota}, $f$ must be an \SDQ{} mechanism for some ordering $\sigma$ and some set of quotas $(q_1,\dots,q_n)$ such that $\sum_{i=1}^n q_i = m$. If $m < n$, the claim follows easily from \Cref{thm:EFX_PO}, so we can assume, without loss of generality, that $m \geq n$. Then, by \Cref{prop:efx_property}, we must have that $q_i \geq 1$ for all $i \in [n]$. Therefore, it suffices to show that $q_i = 1$ for all $i \in [n-1]$.

Assume, without loss of generality, that $\sigma = (1,2,\dots,n)$. Consider a preference profile $\>$ with \emph{identical} preferences, i.e., $\>_i = \>_k$ for all $i,k \in [n]$. Let $g_1 \succ_i g_2 \succ_i \dots \succ_i g_m$ for any $i \in [n]$. Suppose, for contradiction, that $q_i > 1$ for some $i \in [n-1]$, and let $k \in [n-1]$ be the smallest index for which this happens. Since $f$ is an \SDQ{} mechanism, we have that $g_k \in f_k(\>)$ and $|f_k(\>)| = q_k > 1$. Then, for every $\ell > k$, agent $\ell$ envies agent $k$. By \Cref{prop:efx_property}, $f$ violates \EFX{}, which is a contradiction. Therefore, $f$ must be identical to Algorithm~\ref{alg:IQSD} for the ordering $\sigma$, as desired.
\end{proof}

\begin{restatable}{remark}{groupstrategyproof}
We note that any deterministic strategyproof and non-bossy mechanism is also group-strategyproof~\citep{Papai00:Strategyproof}. Therefore, Algorithm~\ref{alg:IQSD} also characterizes the set of \EFX{}, \PO{}, group-strategyproof, non-bossy, and neutral mechanisms under lexicographic preferences.
\label{rmk:groupstrategyproof}
\end{restatable}

In \Cref{prop:minimalityGoods} (whose proof is presented in \Cref{sec:app:prop:minimalityGoods} of the appendix), we show that the set of properties considered in \Cref{thm:EFX_PO_SP_Neutral_NonBossy} is \emph{minimal}. That is, dropping any property from the characterization necessarily allows for feasible mechanisms beyond those in Algorithm~\ref{alg:IQSD}.

\begin{restatable}{prop}{minimalityGoods}
The set $\{$\EFX{}, \PO{}, strategyproofness, non-bossiness, neutrality$\}$ is a minimal set of properties for characterizing the family of mechanisms in Algorithm~\ref{alg:IQSD}.
\label{prop:minimalityGoods}
\end{restatable}

The efficiency guarantee in \Cref{thm:EFX_PO_SP_Neutral_NonBossy} cannot be strengthened much further, as there exists an instance where any \emph{rank-maximal} (\RM{}) mechanism violates strategyproofness~(\Cref{eg:SP+RM_NonExistence}).

\begin{example}[\textbf{Strategyproofness and \RM{}}]
Consider the instance below with $k+2$ goods $g_1,\dots,g_{k+2}$ and three agents:
\begin{align*}
a_1: ~&~ g_1 \, \> \, g_2 \, \> \, g_3 \, \> \dots \> \, g_{k+1} \, \> \, g_{k+2}\\ \nonumber
a_2: ~&~ g_1 \, \> \, g_2 \, \> \, g_3 \, \> \dots \> \, g_{k+1} \, \> \, g_{k+2}\\ \nonumber
a_3: ~&~ g_{2} \, \> \, g_3 \, \> \, g_4 \, \> \dots \> \, g_{k+2} \, \> \, g_1. \nonumber
\end{align*}
Each of the goods $g_2,\dots,g_{k+2}$ is ranked higher by $a_3$ than by $a_1$ or $a_2$, and therefore must be assigned to $a_3$ in any rank-maximal allocation. Suppose, under truthful reporting, $g_1$ is assigned to $a_1$, and $a_2$ gets an empty bundle. Then, $a_2$ could falsely report $g_3$ as its favorite good. By rank-maximality, $g_3$ is now assigned to $a_2$, resulting in a strict improvement.
\label{eg:SP+RM_NonExistence}
\end{example}

The non-existence result in \Cref{eg:SP+RM_NonExistence} prompts us to forego strategyproofness (as well as non-bossiness and neutrality) and focus only on (approximate) envy-freeness and rank-maximality.

\section{Envy-Freeness and Rank-Maximality}

For lexicographic preferences, it is easy to see that a complete allocation is envy-free if and only if each agent receives its favorite good. Checking the existence of an envy-free allocation therefore boils down to computing a (left-)perfect matching in a bipartite graph where the left and the right vertex sets correspond to the agents and the goods, respectively, and the edges denote the top-ranked good of each agent. If an envy-free partial allocation of the top-ranked goods exists, then it can be extended to a complete rank-maximal allocation by assigning each remaining good to an agent that has the highest rank for it (note that the assignment of the remaining goods does not introduce any envy). Thus, the existence of an envy-free and rank-maximal allocation can be checked efficiently for lexicographic preferences (\Cref{prop:EF_RM_Polytime}).

\begin{restatable}[]{prop}{EFRM}
There is a polynomial-time algorithm that, given a lexicographic instance as input, computes an envy-free and rank-maximal allocation, whenever one exists.
\label{prop:EF_RM_Polytime}
\end{restatable}

Since an envy-free allocation is not guaranteed to exist, one could ask whether rank-maximality can always be achieved alongside \emph{approximate} envy-freeness; in particular, \EF{k} and \EFX{}. \Cref{eg:EFk+RM_NonExistence} shows that both of these notions could conflict with rank-maximality. Specifically, for any fixed $k \in \mathbb{N}$, an \EF{k}+\RM{} allocation could fail to exist. Since \EFX{} implies \EF{1}, a similar incompatibility holds for \EFX{}+\RM{} as well.

\begin{example}[\textbf{\EF{k} and \RM{}}]
Consider again the instance in \Cref{eg:SP+RM_NonExistence}. Any rank-maximal allocation assigns the goods $g_2,\dots,g_{k+2}$ to $a_3$. If $g_1$ is assigned to $a_1$, then $a_2$ gets an empty bundle and the pair $\{a_2,a_3\}$ violates \EF{k}.
\label{eg:EFk+RM_NonExistence}
\end{example}

Given the non-existence result in \Cref{eg:EFk+RM_NonExistence}, a natural question is whether there exists an efficient algorithm for checking the existence of an approximately envy-free and rank-maximal allocation. Unfortunately, the news here is also negative, as the problem turns out to be \NPC{} (\Cref{thm:EFX_RM_NP-complete}). Thus, while \EFX{} can always be achieved in conjunction with Pareto optimality (\Cref{thm:EFX_PO,thm:EFX_PO_SP_Neutral_NonBossy}), strengthening the efficiency notion to rank-maximality results in non-existence and computational hardness.

\begin{restatable}[]{thm}{EFXRM}
Determining whether a given instance admits an \EFX{} and rank-maximal allocation is \NPC{}.
\label{thm:EFX_RM_NP-complete}
\end{restatable}

\begin{proof} 
Membership in \NP{} follows from the fact that both \EFX{} and rank-maximality can be checked in polynomial time. To prove \NPH{}ness, we will show a reduction from a restricted version of \ThreeSAT{} called \TTTSAT{}, which is known to be \NPC{}~\citep{AD19sat}. An instance of \TTTSAT{} consists of a collection of $r$ variables $X_1,\dots,X_r$ and $s$ clauses $C_1,\dots,C_s$, where each clause is specified as a disjunction of three literals, and each variable occurs in exactly four clauses, twice negated and twice non-negated. The goal is to determine if there is a truth assignment that satisfies all clauses.

\emph{Construction of the reduced instance}: We will construct a fair division instance with $n = 4r$ agents and $m=4r+s$ goods. The set of agents consists of $2r$ \emph{literal} agents $\{x_i,\overline{x}_i\}_{i \in [r]}$, and $2r$ \emph{dummy} agents $\{d_i,\overline{d}_i\}_{i \in [r]}$. The set of goods consists of $2r$ \emph{signature} goods $\{S_i,\overline{S}_i\}_{i \in [r]}$, $s$ \emph{clause} goods $\{C_j\}_{j \in [s]}$, and $2r$ \emph{dummy} goods $\{T_i,B_i\}_{i \in [r]}$; here $T_i$ and $B_i$ denote the \emph{top} and the \emph{bottom} dummy goods associated with the variable $X_i$, respectively.

\begin{table}[ht]
\renewcommand{\arraystretch}{1.35}
\centering
    \begin{tabular}{|cl|}
    \hline
    $\vartriangleright$: & $S_1 \succ \overline{S}_1 \succ \dots \succ S_r \succ \overline{S}_r \succ T_1 \succ \dots \succ T_r$ \hfill $\succ C_1 \succ \dots \succ C_s \succ B_1 \succ \dots \succ B_r$ \\
    \hline
    $x_i$: & $S_i \succ \, \vartriangleright_{(j-1)} \, \succ C_j \succ \,  \vartriangleright_{(k-j-1)} \, \succ C_k \succ *$ \\
    \hline
    $\overline{x}_i$: & $\overline{S}_i \succ \, \vartriangleright_{(p-1)} \, \succ C_p \succ \,  \vartriangleright_{(q-p-1)} \, \succ C_q \succ *$ \\
    \hline
    $d_i$: & $T_i \succ S_i \succ B_i \succ * ~~~ \text{ and } ~~~ \overline{d}_i: T_i \succ \overline{S}_i \succ B_i \succ *.$ \\
    \hline
    \end{tabular}
    \caption{Preferences of agents in the proof of \Cref{thm:EFX_RM_NP-complete}.}
    \label{tab:EFX_RM_NP-complete}
\end{table}

\emph{Preferences}: \Cref{tab:EFX_RM_NP-complete} shows the preferences of the agents. Let $\vartriangleright$ define a \emph{reference} ordering on the set of goods. For every $i \in [r]$, if $C_{j}$ and $C_{k}$ denote the two clauses containing the positive literal $x_i$, then the literal agent $x_i$ ranks $S_i$ at the top, and the clause goods $C_j$ and $C_k$ at ranks $j+1$ and $k+1$, respectively. The missing positions consist of remaining goods ranked according to $\vartriangleright$ (we write $\vartriangleright_{\ell}$ to denote the top $\ell$ goods in $\vartriangleright$ that have not been ranked so far). The symbol $*$ indicates rest of the goods ordered according to $\vartriangleright$. The preferences of the (negative) literal agent $\overline{x}_i$ and the dummy agents $d_i$, $\overline{d}_i$ are defined similarly as shown in \Cref{tab:EFX_RM_NP-complete}. This completes the construction of the reduced instance.

Note that for any fixed $i \in [r]$, the signature good $S_i$ (or $\overline{S}_i$) is ranked at the top position by the literal agent $x_i$ (or $\overline{x}_i$), and at a lower position by all other agents. Therefore, any rank-maximal allocation must assign $S_i$ to $x_i$ and $\overline{S}_i$ to $\overline{x}_i$. For a similar reason, a rank-maximal allocation must assign the clause good $C_j$ to a literal agent corresponding to a literal contained in the clause $C_j$, and the dummy goods $T_i,B_i$ to the dummy agents $d_i,\overline{d}_i$. The aforementioned \emph{necessary} conditions for rank-maximality are also \emph{sufficient} since each clause good $C_j$ is ranked at the same position by all literal agents corresponding to the literals contained in clause $C_j$, and the goods $T_i$ and $B_i$ are ranked identically by $d_i$ and $\overline{d}_i$.

We will now argue the equivalence of solutions.

($\Rightarrow$) Given a satisfying truth assignment, the desired allocation, say $A$, can be constructed as follows: For every $i \in [r]$, assign the signature goods $S_i$ and $\overline{S}_i$ to the literal agents $x_i$ and $\overline{x}_i$, respectively. If $x_i = 1$, then assign $T_i$ to $d_i$ and $B_i$ to $\overline{d}_i$, otherwise, if $x_i = 0$, then assign $T_i$ to $\overline{d}_i$ and $B_i$ to $d_i$. For every $j \in [s]$, the clause good $C_j$ is assigned to a literal agent $x_i$ (or $\overline{x}_i$) if the literal $x_i$ (or $\overline{x}_i$) is contained in the clause $C_j$ and the clause is satisfied by the literal, i.e., $x_i = 1$ (or $\overline{x}_i = 1$). Note that under a satisfying assignment, each clause must have at least one such literal.

Observe that allocation $A$ satisfies the aforementioned sufficient condition for rank-maximality. Furthermore, any envied agent in $A$ receives exactly one good; in particular, if $d_i$ receives a bottom dummy good $B_i$, then we have $x_i = 0$ in which case the literal agent $x_i$, who is envied by $d_i$, does not receive any clause goods. By \Cref{prop:efx_property}, $A$ is \EFX{}.

($\Leftarrow$) Now suppose there exists an \EFX{} and rank-maximal allocation $A$. Then, $A$ must satisfy the aforementioned necessary condition for rank-maximality. That is, for every $i \in [r]$, the signature goods $S_i$ and $\overline{S}_i$ are assigned to the literal agents $x_i$ and $\overline{x}_i$, respectively (i.e., $S_i \in A_{x_i}$ and $\overline{S}_i \in A_{\overline{x}_i}$), and the dummy goods $T_i$ and $B_i$ are allocated between the dummy agents $d_i$ and $\overline{d}_i$ (i.e., $\{T_i,B_i\} \subseteq A_{d_i} \cup A_{\overline{d}_i}$). In addition, for every $j \in [s]$, the clause good $C_j$ is assigned to a literal agent $x_i$ (or $\overline{x}_i$) such that the literal $x_i$ (or $\overline{x}_i$) is contained in the clause $C_j$. Also, by \Cref{prop:efx_property}, each dummy agent must get exactly one dummy good.

We will construct a truth assignment for the \TTTSAT{} instance as follows: For every $i \in [r]$, if $T_i \in A_{d_i}$, then set $x_i = 1$, otherwise set $x_i = 0$. Note that the assignment is feasible as no literal is assigned conflicting values. To see why this is a satisfying assignment, consider any clause $C_j$. Suppose the clause good $C_j$ is assigned to a literal agent $x_i$ (an analogous argument works when $\overline{x}_i$ gets $C_j$). Then, due to rank-maximality, we know that the literal $x_i$ must be contained in the clause $C_j$. Furthermore, since agent $x_i$ gets more than one good ($S_i,C_j \in A_{x_i}$), it cannot be envied under $A$ (\Cref{prop:efx_property}). Thus, the dummy agent $d_i$ must get the top good $T_i$. Recall that in this case we set $x_i=1$. Since clause $C_j$ contains $x_i$, it must be satisfied, as desired.
\end{proof}

The intractability in \Cref{thm:EFX_RM_NP-complete} persists even when we relax the fairness requirement from \EFX{} to \EF{k}.

\begin{restatable}[]{thm}{EFkRM}
For any fixed $k \in \mathbb{N}$, determining the existence of an \EF{k} and rank-maximal allocation is \NPC{}.
\label{thm:EFk_RM_NP-complete}
\end{restatable}

We note that the proof of \Cref{thm:EFk_RM_NP-complete} (see \Cref{sec:app:thm:EFk_RM_NP-complete} of the appendix) differs considerably from that of \Cref{thm:EFX_RM_NP-complete} as neither result is an immediate consequence of the other. Indeed, a YES instance of \EFX{}+\RM{} is also a YES instance of \EF{k}+\RM{}, but the same is not true for a NO instance.

A corollary of \Cref{thm:EFk_RM_NP-complete} is that checking the existence of \EF{1}+\RM{} allocations for \emph{additive valuations} is also \NPC{}.\footnote{An additive valuations instance in which agent $i$ values its $j^\textup{th}$ favorite good at $2^{m-j+1}$ is equivalent to the lexicographic instance.} For this setting, \citet{AHM+19constrained} have shown \NPC{}ness even for three agents. By contrast, we will show that for lexicographic preferences, the problem is efficiently solvable when $n=3$ (\Cref{prop:EF1_RM_ThreeAgents}). The proof of this result is deferred to \Cref{sec:app:prop:EF1_RM_ThreeAgents} of the appendix.

\begin{restatable}[]{prop}{EFoneRMThreeAgents}
There is a polynomial-time algorithm that, given as input a lexicographic instance with three agents, computes an \EF{1} and rank-maximal allocation, whenever one exists.
\label{prop:EF1_RM_ThreeAgents}
\end{restatable}

Another avenue for circumventing the intractability in \Cref{thm:EFX_RM_NP-complete} is provided by \emph{maximin share guarantee} (\MMS{}). For additive valuations, \EFX{} and \MMS{} are incomparable notions~\citep{ABM18comparing}. However, for lexicographic preferences, \MMS{} is strictly weaker than \EFX{} (see \Cref{prop:EFX_implies_MMS} in \Cref{sec:app:prop:EFX_implies_MMS} of the appendix). This relaxation of \EFX{} turns out to be computationally useful, as the existence of an \MMS{} and rank-maximal allocations can be checked in polynomial time.

\begin{restatable}[]{thm}{MMSRMGoods}
There is a polynomial-time algorithm that, given as input a lexicographic instance, computes an \MMS{} and rank-maximal allocation, whenever one exists.
\label{thm:MMS_RM_Goods}
\end{restatable}
\begin{proof}
Fix any agent $i \in N$, and suppose its preference is given by $\>_i \coloneqq g_1 \, \> \, g_2 \, \> \dots \> \, g_m$. Under lexicographic preferences, the \MMS{} partition of agent $i \in N$ is uniquely defined as
$$\left\{\{g_1\},\{g_2\},\dots,\{g_{n-1}\},\{g_n,g_{n+1},\dots,g_m\}\right\}.$$
 This observation gives a characterization of \MMS{} allocations: An allocation is \MMS{} if and only if each agent either receives one or more of its top-$(n-1)$ goods, or it receives all of its bottom-$(m-n+1)$ goods.

Construct a bipartite graph $G = (N \cup M, E)$ between agents and goods where an edge $(i,j) \in E$ exists if agent $i$ ranks good $j$ within its top-$(n-1)$ goods, and good $j$ can be `rank-maximally assigned' to agent $i$ (in other words, agent $i$ ranks good $j$ at least as high as any other agent).

If $G$ admits a perfect matching (this can be checked in polynomial time), then, by the above characterization, we have a partial allocation that is \MMS{} and rank-maximal. By assigning the unmatched goods in a rank-maximal manner, we obtain a desired complete allocation.

Thus, for the rest of the proof, let us assume that $G$ does not have a perfect matching. Note that in this case, the unordered set of top-$(n-1)$ goods of each agent is the same. Then, either there is no \MMS{}+rank-maximal allocation, or if there exists one, then it must be that some agent gets \emph{all} of its bottom-$(m-n+1)$ goods. The latter condition can be checked in polynomial time as follows (note that this would establish the desired polynomial running time of our algorithm): First, we check for each agent whether it can be assigned its bottom-$(m-n+1)$ goods in a rank-maximal manner; let $S$ denote the set of all agents who satisfy this condition. If $S$ is empty, then we can return NO. Otherwise, we move to the next step. 

Fix an arbitrary agent $i \in S$, and let $g_n,g_{n+1},...,g_m$ denote its bottom-$(m-n+1)$ goods. Since each of $g_n,g_{n+1},...,g_m$ is rank-maximal for agent $i$, it must be that \emph{every} agent ranks these goods in the \emph{exact same way}, i.e., between the positions $n$ and $m$. The reason is that since the good $g_m$ can be rank-maximally assigned to agent $i$, it must be ranked last by \emph{every} agent. Subject to that, the good $g_{m-1}$ must be ranked at $m-1$ by every agent, and so on. In this case, we can ``bundle up'' the bottom $m-n+1$ goods into a single ``meta'' good, and treat the new instance as one with $n$ agent and $n$ goods, where the meta good is ranked last by everybody. Then, the old instance admits an \MMS{} allocation if and only if the new instance admits one, which, in turn, happens if and only if the bipartite graph of the new instance admits a perfect matching.
\end{proof}

\section{Experiments}
We now revisit the non-existence result in \Cref{eg:EFk+RM_NonExistence} by examining how frequently fair (i.e., \EF{}, \EFX{}, \EF{1}, \MMS{}) and efficient (i.e., rank-maximal) allocations exist in synthetically generated data. To that end, we consider a fixed number of agents ($n=5$) whose preferences over a set of $m$ goods (where $m \in \{5,\dots,100\}$) are generated using the Mallows model~\citep{Mallows1957}. Given a reference ranking $\>^* \in \L$ and a dispersion parameter $\phi \in [0,1]$, the probability of generating a ranking $\>_i \in \L$ under the Mallows model is given by $\frac{1}{Z} \phi^{\texttt{d}(\>^*,\>_i)}$, where $Z$ is a normalization constant and $\texttt{d}(\cdot)$ is the Kendall's Tau distance. Thus, $\phi=0$ induces identical preferences (i.e., $\>_i=\>^*$) while $\phi=1$ is the uniform distribution. For each combination of $m$, $n$, and $\phi\in\{0,0.25,0.5,0.75,1\}$, we sample $1000$ preference profiles, and use an integer linear program to check the existence of $\{\EF{},\EFX{},\EF{1},\MMS{}\}+\RM{}$ allocations. Code and data for all our experiments is available at \url{https://github.com/sujoyksikdar/Envy-Freeness-With-Lexicographic-Preferences}.

\begin{figure*}[ht]
    \centering
        \includegraphics[width=0.49\textwidth]{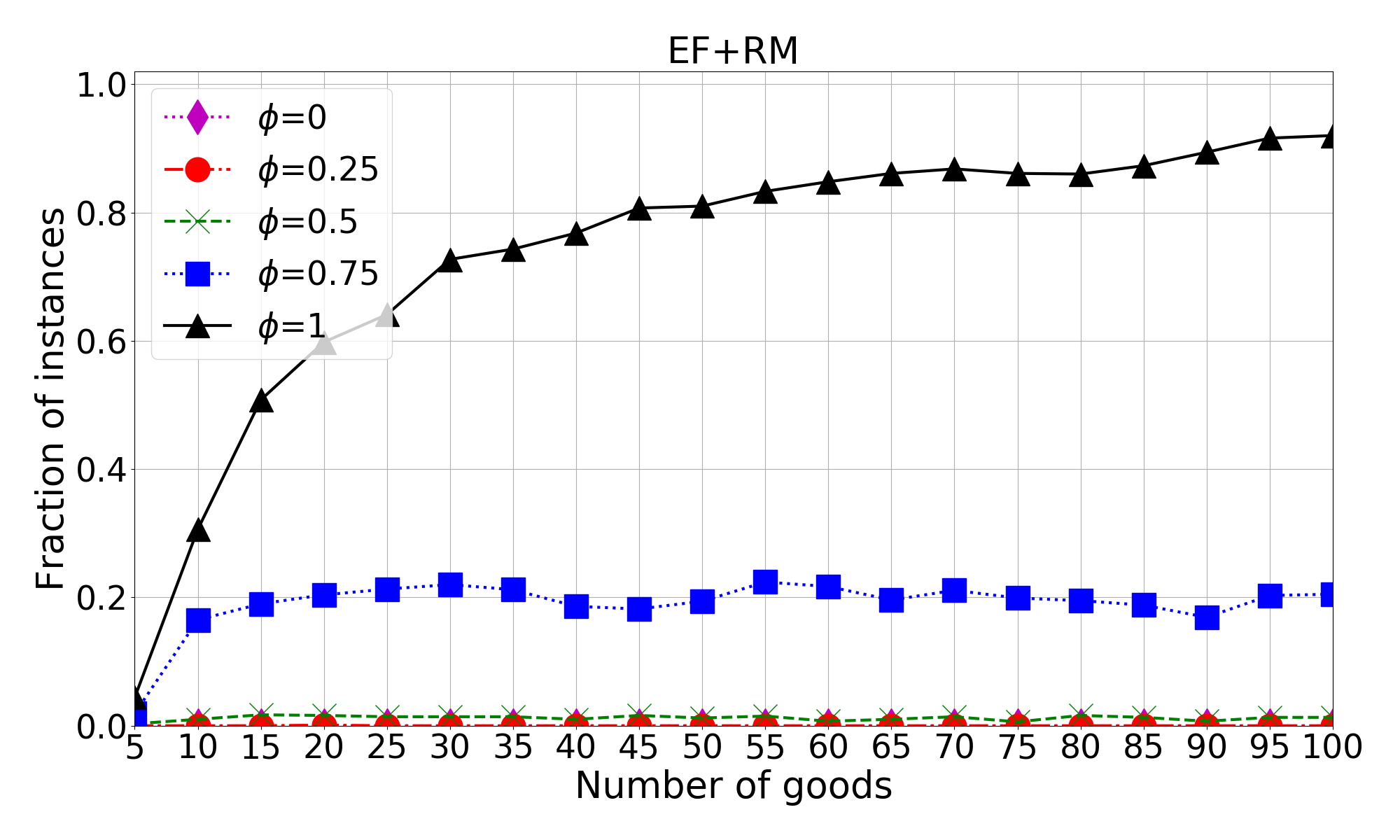}
        \includegraphics[width=0.49\textwidth]{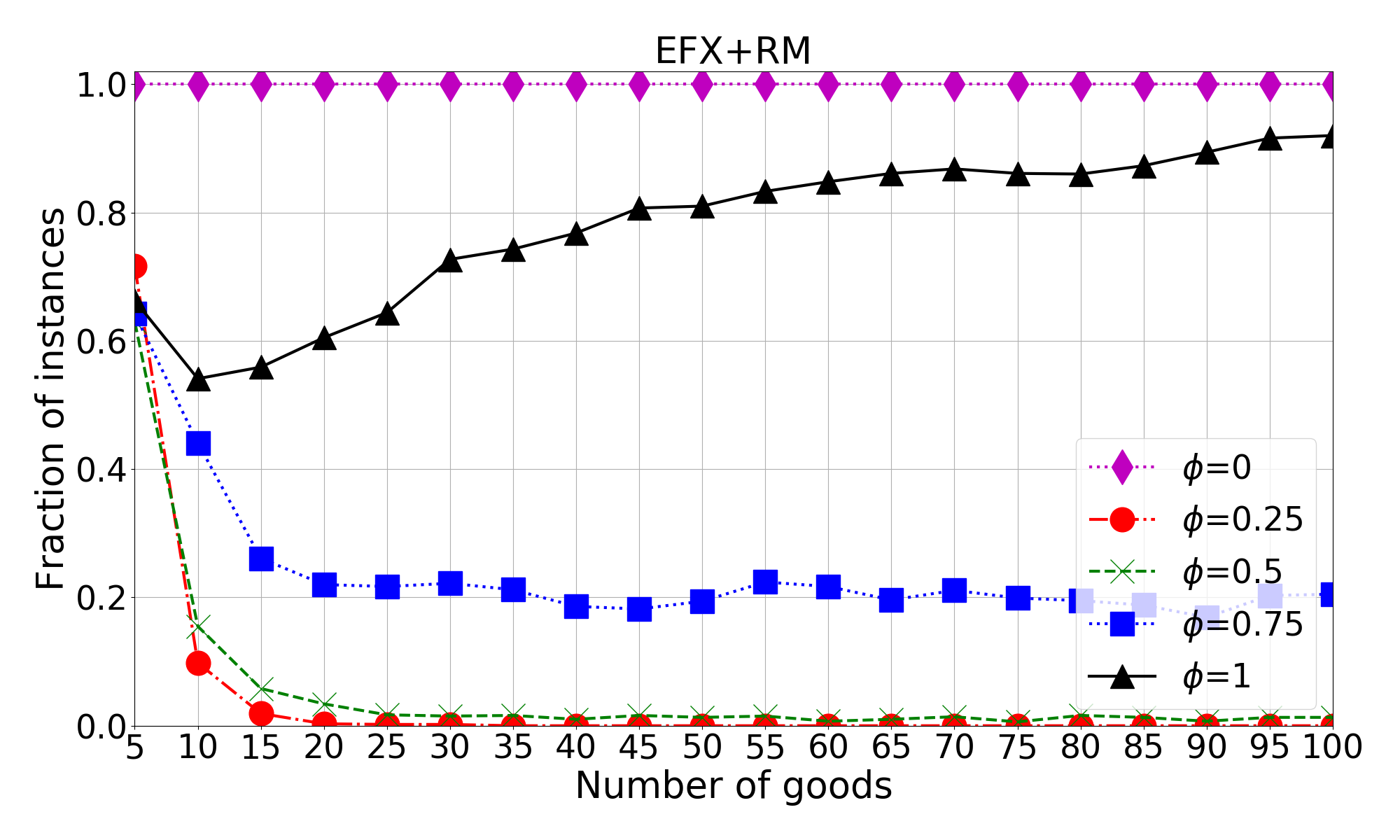}\\
        \includegraphics[width=0.49\textwidth]{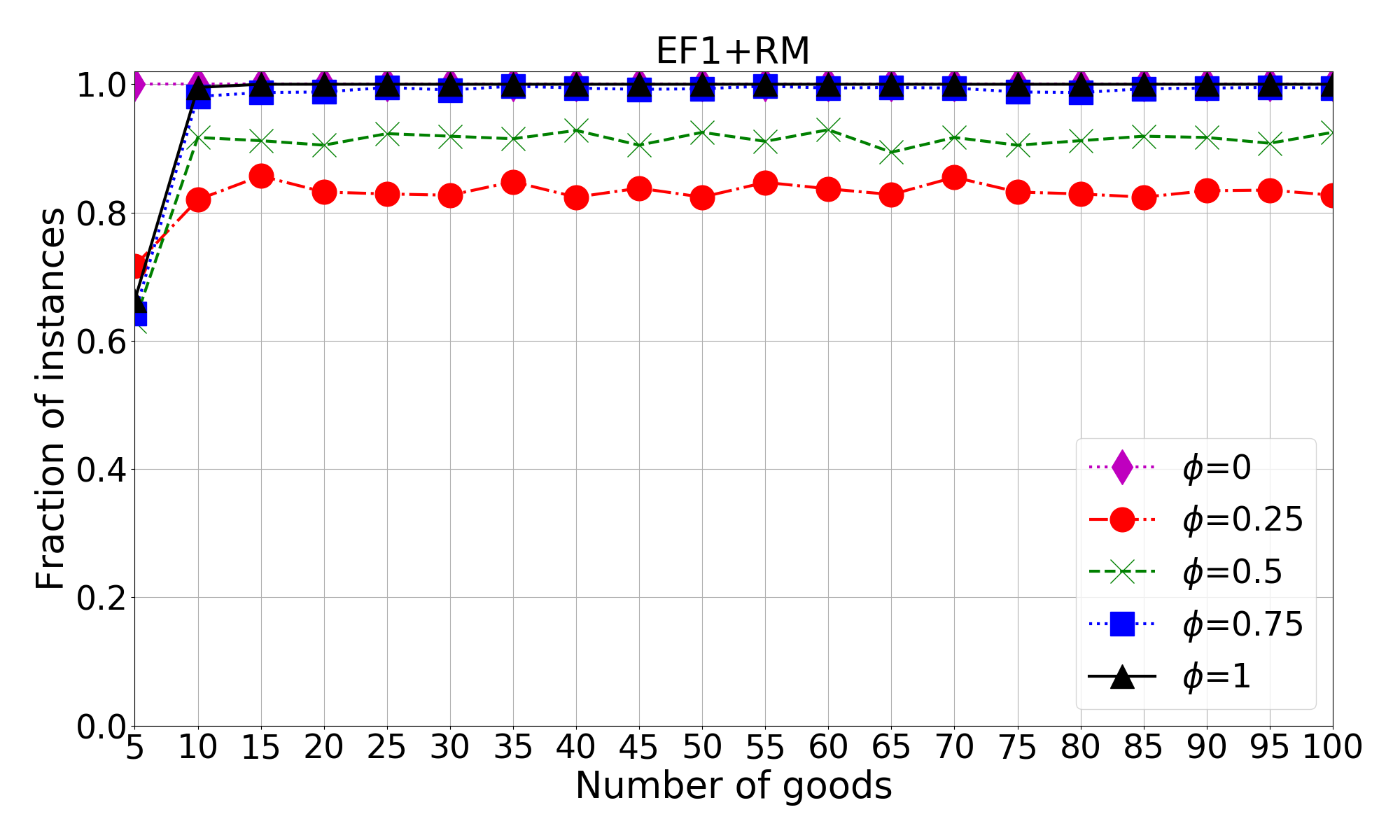}
        \includegraphics[width=0.49\textwidth]{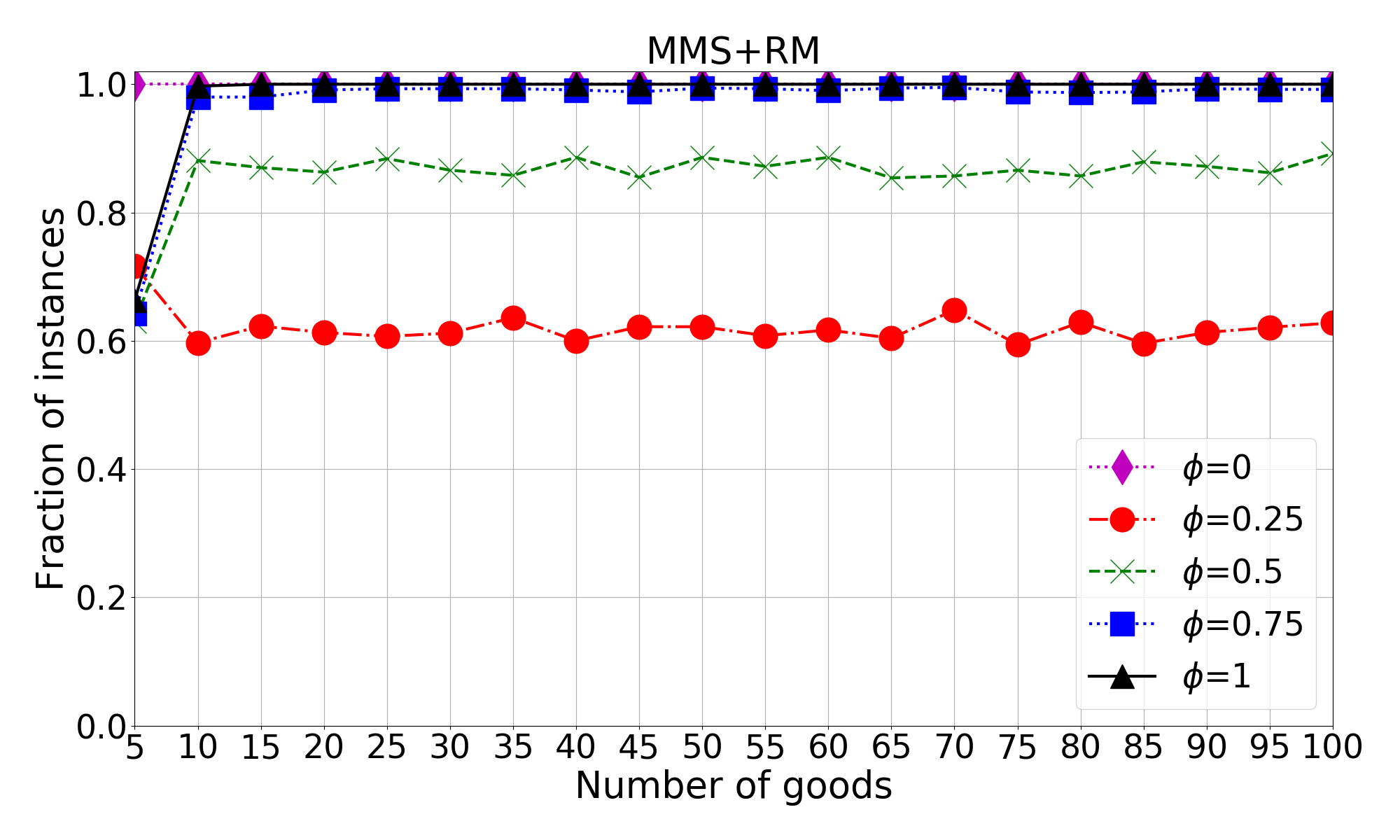}
    \caption{The plots show how the fraction of instances that admit $\{\EF{},\EFX{},\EF{1},\MMS{}\}+\RM{}$ allocations varies with the number of goods. The number of agents is fixed ($n=5$), and their preferences follow the Mallows model with dispersion parameter $\phi$.}
    \label{fig:mallows}
\end{figure*}

\Cref{fig:mallows} presents our experimental results. For identical preferences ($\phi=0$), every complete allocation is Pareto optimal as well as rank-maximal. Therefore, an \EFX{}+\RM{} (and hence \{\EF{1}, \MMS{}\}+\RM{}) allocation always exists in this case, validating our theoretical result in \Cref{thm:EFX_PO}. On the other hand, an \EF{}+\RM{} allocation fails to exist because of the conflict in top-ranked goods. At the other extreme for $\phi=1$ (i.e., the uniform distribution), we note that the probability of existence of \EF{}+\RM{} outcomes grows steadily with $m$. This is because for (exact) envy-freeness, all five rankings should have distinct top goods, the probability of which is $(1-\frac{1}{m}) \cdot (1-\frac{2}{m}) \cdot (1-\frac{3}{m}) \cdot (1-\frac{4}{m})$. For $m=100$, this value is roughly $0.9$, suggesting that in the asymptotic regime, envy-free (and, by extension, envy-free and rank-maximal) allocations are increasingly likely to exist (as \Cref{fig:mallows} shows), and that our algorithm in \Cref{prop:EF_RM_Polytime} will return \EF{}+\RM{} outcomes with high probability.

We observe that for small values of $m$, a significantly greater fraction of instances admit \EFX{}+\RM{} allocations than \EF{}+\RM{} allocations. However, as the number of goods increases, this gap, i.e. the fraction of instances that do not admit \EF{}+\RM{} allocations but do admit \EFX{}+\RM{} allocations, shrinks rapidly. We conjecture that the likelihood that every agent must be allocated more than one good in any \RM{} allocation increases as the number of goods increases. Therefore, it is likely that envied agents receive more than one good, which is in direct conflict with \EFX{}. As we show in \Cref{sec:app:experiments} of the appendix, for relatively small values of $m$, the fraction of instances that admit \EFX{}+\RM{} allocations decreases initially with the number of goods up to a point, after which the difference between the fractions of instances that admit \EFX{}+\RM{} and \EF{}+\RM{} becomes negligible, and we observe an increasing trend both in the fraction of instances that admit \EF{}+\RM{} allocations and those that admit \EFX{}+\RM{} allocations.

We also observe a general trend in our plots that $\{\EF1{},\MMS{}\}+\RM{}$ allocations tend to exist more frequently as the number of goods increases. This together with the similar increasing trend for \EF{}+\RM{} and \EFX{}+\RM{}, suggests that the distributional approach could be a promising avenue for addressing the non-existence result in \Cref{eg:EFk+RM_NonExistence}.

\section{Concluding Remarks}

We studied the interplay of fairness and efficiency under lexicographic preferences, obtaining strong algorithmic characterizations for \EFX{} and Pareto optimality that addressed notable gaps in the additive valuations model, and outlining the computational limits of our approach for the stronger efficiency notion of rank-maximality.

A natural extension to our preference model is including ties or weak orders where lexicographic preferences are over `equivalence classes'. Clearly, the computational intractability and incompatibility results in this paper extend to this larger class.
When preferences can include ties, several intricate axiomatic and technical challenges may arise. An interesting future direction is to start from mechanisms proposed by \citet{krysta2014size} and \citet{bogomolnaia2005strategy} for achieving strategyproof and \PO{} allocations, and investigate the existence and compatibility of \EFX{} along with other desirable properties.

Going forward, it would be interesting to develop distribution-specific algorithms that can compute \EFX{}+\RM{} (or \EF{1}+\RM{}) allocations with, say, a constant probability. Extending our algorithmic characterization results to other fairness notions, specifically \EF{1} and \MMS{}, will also be of interest.

\section*{Acknowledgments}
HH acknowledges support from NSF grant \#1850076. RV acknowledges support from ONR\#N00014-171-2621 while he was affiliated with Rensselaer Polytechnic Institute, and is currently supported by project no. RTI4001 of the Department of Atomic Energy, Government of India. Part of this work was done while RV was supported by the Prof. R Narasimhan postdoctoral award. LX acknowledges support from NSF \#1453542 and \#1716333,  ONR \#N00014-171-2621, and a gift fund from Google. We thank the anonymous reviewers for their very helpful comments and suggestions.

\bibliographystyle{named}
\bibliography{references,ref}

\begin{thebibliography}{}

\bibitem[\protect\citeauthoryear{Ahadi and Dehghan}{2019}]{AD19sat}
Arash Ahadi and Ali Dehghan.
\newblock {(2/2/3)-SAT Problem and its Applications in Dominating Set
  Problems}.
\newblock {\em Discrete Mathematics \& Theoretical Computer Science}, 21(4),
  2019.

\bibitem[\protect\citeauthoryear{Amanatidis \bgroup \em et al.\egroup
  }{2017}]{amanatidis2017truthful}
Georgios Amanatidis, Georgios Birmpas, George Christodoulou, and Evangelos
  Markakis.
\newblock {Truthful Allocation Mechanisms without Payments: Characterization
  and Implications on Fairness}.
\newblock In {\em Proceedings of the 2017 ACM Conference on Economics and
  Computation}, pages 545--562, 2017.

\bibitem[\protect\citeauthoryear{Amanatidis \bgroup \em et al.\egroup
  }{2018}]{ABM18comparing}
Georgios Amanatidis, Georgios Birmpas, and Evangelos Markakis.
\newblock {Comparing Approximate Relaxations of Envy-Freeness}.
\newblock In {\em Proceedings of the 27th International Joint Conference on
  Artificial Intelligence}, pages 42--48, 2018.

\bibitem[\protect\citeauthoryear{Amanatidis \bgroup \em et al.\egroup
  }{2020a}]{ABF+20maximum}
Georgios Amanatidis, Georgios Birmpas, Aris Filos-Ratsikas, Alexandros
  Hollender, and Alexandros~A Voudouris.
\newblock {Maximum Nash Welfare and Other Stories about EFX}.
\newblock In {\em In Proceedings of the 29th International Joint Conference on
  Artificial Intelligence}, pages 24--30, 2020.

\bibitem[\protect\citeauthoryear{Amanatidis \bgroup \em et al.\egroup
  }{2020b}]{AMN20multiple}
Georgios Amanatidis, Evangelos Markakis, and Apostolos Ntokos.
\newblock {Multiple Birds with One Stone: Beating $1/2$ for EFX and GMMS via
  Envy Cycle Elimination}.
\newblock {\em Theoretical Computer Science}, 841:94 -- 109, 2020.

\bibitem[\protect\citeauthoryear{Aziz and
  Kasajima}{2017}]{aziz2017impossibilities}
Haris Aziz and Yoichi Kasajima.
\newblock {Impossibilities for Probabilistic Assignment}.
\newblock {\em Social Choice and Welfare}, 49(2):255--275, 2017.

\bibitem[\protect\citeauthoryear{Aziz \bgroup \em et al.\egroup
  }{2019}]{AHM+19constrained}
Haris Aziz, Xin Huang, Nicholas Mattei, and Erel Segal-Halevi.
\newblock {The Constrained Round Robin Algorithm for Fair and Efficient
  Allocation}.
\newblock {\em arXiv preprint arXiv:1908.00161}, 2019.

\bibitem[\protect\citeauthoryear{Barman \bgroup \em et al.\egroup
  }{2018}]{BKV18finding}
Siddharth Barman, Sanath~Kumar Krishnamurthy, and Rohit Vaish.
\newblock {Finding Fair and Efficient Allocations}.
\newblock In {\em Proceedings of the 2018 ACM Conference on Economics and
  Computation}, pages 557--574, 2018.

\bibitem[\protect\citeauthoryear{Baumeister \bgroup \em et al.\egroup
  }{2017}]{BBL+17positional}
Dorothea Baumeister, Sylvain Bouveret, J{\'e}r{\^o}me Lang, Nhan-Tam Nguyen,
  Trung~Thanh Nguyen, J{\"o}rg Rothe, and Abdallah Saffidine.
\newblock {Positional Scoring-Based Allocation of Indivisible Goods}.
\newblock {\em Autonomous Agents and Multi-Agent Systems}, 31(3):628--655,
  2017.

\bibitem[\protect\citeauthoryear{B{\'e}rczi \bgroup \em et al.\egroup
  }{2020}]{BBB+20envy}
Krist{\'o}f B{\'e}rczi, Erika~R B{\'e}rczi-Kov{\'a}cs, Endre Boros,
  Fekadu~Tolessa Gedefa, Naoyuki Kamiyama, Telikepalli Kavitha, Yusuke
  Kobayashi, and Kazuhisa Makino.
\newblock {Envy-free Relaxations for Goods, Chores, and Mixed Items}.
\newblock {\em arXiv preprint arXiv:2006.04428}, 2020.

\bibitem[\protect\citeauthoryear{Bogomolnaia and
  Moulin}{2001}]{Bogomolnaia01:New}
Anna Bogomolnaia and Herv\'e Moulin.
\newblock {A New Solution to the Random Assignment Problem}.
\newblock {\em Journal of Economic Theory}, 100(2):295--328, 2001.

\bibitem[\protect\citeauthoryear{Bogomolnaia \bgroup \em et al.\egroup
  }{2005}]{bogomolnaia2005strategy}
Anna Bogomolnaia, Rajat Deb, and Lars Ehlers.
\newblock {Strategy-Proof Assignment on the Full Preference Domain}.
\newblock {\em Journal of Economic Theory}, 123(2):161--186, 2005.

\bibitem[\protect\citeauthoryear{Budish}{2011}]{budish2011combinatorial}
Eric Budish.
\newblock {The Combinatorial Assignment Problem: Approximate Competitive
  Equilibrium from Equal Incomes}.
\newblock {\em Journal of Political Economy}, 119(6):1061--1103, 2011.

\bibitem[\protect\citeauthoryear{Caragiannis \bgroup \em et al.\egroup
  }{2019a}]{CGH19envy}
Ioannis Caragiannis, Nick Gravin, and Xin Huang.
\newblock {Envy-Freeness Up to Any Item with High Nash Welfare: The Virtue of
  Donating Items}.
\newblock In {\em Proceedings of the 2019 ACM Conference on Economics and
  Computation}, pages 527--545, 2019.

\bibitem[\protect\citeauthoryear{Caragiannis \bgroup \em et al.\egroup
  }{2019b}]{caragiannis2019unreasonable}
Ioannis Caragiannis, David Kurokawa, Herv{\'e} Moulin, Ariel~D Procaccia,
  Nisarg Shah, and Junxing Wang.
\newblock {The Unreasonable Fairness of Maximum Nash Welfare}.
\newblock {\em ACM Transactions on Economics and Computation}, 7(3):12, 2019.

\bibitem[\protect\citeauthoryear{Chaudhury \bgroup \em et al.\egroup
  }{2020a}]{CGM20efx}
Bhaskar~Ray Chaudhury, Jugal Garg, and Kurt Mehlhorn.
\newblock {EFX Exists for Three Agents}.
\newblock In {\em Proceedings of the Twenty-First ACM Conference on Economics
  and Computation}, pages 1--19, 2020.

\bibitem[\protect\citeauthoryear{Chaudhury \bgroup \em et al.\egroup
  }{2020b}]{CGM20fair}
Bhaskar~Ray Chaudhury, Jugal Garg, and Ruta Mehta.
\newblock {Fair and Efficient Allocations under Subadditive Valuations}.
\newblock {\em arXiv preprint arXiv:2005.06511}, 2020.

\bibitem[\protect\citeauthoryear{Chaudhury \bgroup \em et al.\egroup
  }{2020c}]{CKM+20little}
Bhaskar~Ray Chaudhury, Telikepalli Kavitha, Kurt Mehlhorn, and Alkmini
  Sgouritsa.
\newblock {A Little Charity Guarantees Almost Envy-Freeness}.
\newblock In {\em Proceedings of the Fourteenth Annual ACM-SIAM Symposium on
  Discrete Algorithms}, pages 2658--2672, 2020.

\bibitem[\protect\citeauthoryear{Chen and Liu}{2020}]{CL20fairness}
Xingyu Chen and Zijie Liu.
\newblock {The Fairness of Leximin in Allocation of Indivisible Chores}.
\newblock {\em arXiv preprint arXiv:2005.04864}, 2020.

\bibitem[\protect\citeauthoryear{{\'C}usti{\'c} \bgroup \em et al.\egroup
  }{2015}]{CKW15geometric}
Ante {\'C}usti{\'c}, Bettina Klinz, and Gerhard~J Woeginger.
\newblock {Geometric Versions of the Three-Dimensional Assignment Problem under
  General Norms}.
\newblock {\em Discrete Optimization}, 18:38--55, 2015.

\bibitem[\protect\citeauthoryear{Elkind \bgroup \em et al.\egroup
  }{2016}]{ELP16preference}
Edith Elkind, Martin Lackner, and Dominik Peters.
\newblock {Preference Restrictions in Computational Social Choice: Recent
  Progress}.
\newblock In {\em Proceedings of the Twenty-Fifth International Joint
  Conference on Artificial Intelligence}, pages 4062--4065, 2016.

\bibitem[\protect\citeauthoryear{Farhadi \bgroup \em et al.\egroup
  }{2020}]{FHL+20almost}
Alireza Farhadi, MohammadTaghi Hajiaghayi, Mohamad Latifian, Masoud Seddighin,
  and Hadi Yami.
\newblock {Almost Envy-Freeness, Envy-Rank, and Nash Social Welfare Matchings}.
\newblock {\em arXiv preprint arXiv:2007.07027}, 2020.

\bibitem[\protect\citeauthoryear{Fujita \bgroup \em et al.\egroup
  }{2018}]{FLS18complexity}
Etsushi Fujita, Julien Lesca, Akihisa Sonoda, Taiki Todo, and Makoto Yokoo.
\newblock {A Complexity Approach for Core-Selecting Exchange under
  Conditionally Lexicographic Preferences}.
\newblock {\em Journal of Artificial Intelligence Research}, 63:515--555, 2018.

\bibitem[\protect\citeauthoryear{Garey and Johnson}{1979}]{Garey79:Computers}
Michael Garey and David Johnson.
\newblock {\em {Computers and Intractability}}.
\newblock W. H. Freeman and Company, 1979.

\bibitem[\protect\citeauthoryear{Gigerenzer and
  Goldstein}{1996}]{GG96reasoning}
Gerd Gigerenzer and Daniel~G Goldstein.
\newblock {Reasoning the Fast and Frugal Way: Models of Bounded Rationality}.
\newblock {\em Psychological Review}, 103(4):650, 1996.

\bibitem[\protect\citeauthoryear{Hosseini and
  Larson}{2019}]{hosseini2019multiple}
Hadi Hosseini and Kate Larson.
\newblock {Multiple Assignment Problems under Lexicographic Preferences}.
\newblock In {\em Proceedings of the 18th International Conference on
  Autonomous Agents and MultiAgent Systems}, pages 837--845, 2019.

\bibitem[\protect\citeauthoryear{Irving \bgroup \em et al.\egroup
  }{2006}]{IKM+06rank}
Robert~W Irving, Telikepalli Kavitha, Kurt Mehlhorn, Dimitrios Michail, and
  Katarzyna Paluch.
\newblock {Rank-Maximal Matchings}.
\newblock {\em ACM Transactions on Algorithms}, 2:602--610, 2006.

\bibitem[\protect\citeauthoryear{Kojima}{2009}]{kojima2009random}
Fuhito Kojima.
\newblock {Random Assignment of Multiple Indivisible Objects}.
\newblock {\em Mathematical Social Sciences}, 57(1):134--142, 2009.

\bibitem[\protect\citeauthoryear{Krysta \bgroup \em et al.\egroup
  }{2014}]{krysta2014size}
Piotr Krysta, David Manlove, Baharak Rastegari, and Jinshan Zhang.
\newblock {Size Versus Truthfulness in the House Allocation Problem}.
\newblock In {\em Proceedings of the fifteenth ACM conference on Economics and
  computation}, pages 453--470. ACM, 2014.

\bibitem[\protect\citeauthoryear{Lang \bgroup \em et al.\egroup
  }{2018}]{LMX18voting}
J{\'e}r{\^o}me Lang, J{\'e}r{\^o}me Mengin, and Lirong Xia.
\newblock {Voting on Multi-Issue Domains with Conditionally Lexicographic
  Preferences}.
\newblock {\em Artificial Intelligence}, 265:18--44, 2018.

\bibitem[\protect\citeauthoryear{Mahara}{2020}]{M20existence}
Ryoga Mahara.
\newblock {Existence of EFX for Two Additive Valuations}.
\newblock {\em arXiv preprint arXiv:2008.08798}, 2020.

\bibitem[\protect\citeauthoryear{Mallows}{1957}]{Mallows1957}
Colin~L Mallows.
\newblock {Non-Null Ranking Models}.
\newblock {\em Biometrika}, 44(1/2):114--130, 1957.

\bibitem[\protect\citeauthoryear{Nguyen}{2020}]{N20fairly}
Trung~Thanh Nguyen.
\newblock {How to Fairly Allocate Indivisible Resources Among Agents Having
  Lexicographic Subadditive Utilities}.
\newblock In {\em Frontiers in Intelligent Computing: Theory and Applications},
  pages 156--166. 2020.

\bibitem[\protect\citeauthoryear{Paluch}{2013}]{P13capacitated}
Katarzyna Paluch.
\newblock {Capacitated Rank-Maximal Matchings}.
\newblock In {\em International Conference on Algorithms and Complexity}, pages
  324--335. Springer, 2013.

\bibitem[\protect\citeauthoryear{P\'apai}{2000a}]{Papai00:Strategyproof}
Szilvia P\'apai.
\newblock {Strategyproof Assignment by Hierarchical Exchange}.
\newblock {\em Econometrica}, 68(6):1403--1433, 2000.

\bibitem[\protect\citeauthoryear{P\'apai}{2000b}]{Papai00:Strategyproofquotas}
Szilvia P\'apai.
\newblock {Strategyproof Multiple Assignment Using Quotas}.
\newblock {\em Review of Economic Design}, 5:91--105, 2000.

\bibitem[\protect\citeauthoryear{Plaut and Roughgarden}{2020}]{PR20almost}
Benjamin Plaut and Tim Roughgarden.
\newblock {Almost Envy-Freeness with General Valuations}.
\newblock {\em SIAM Journal on Discrete Mathematics}, 34(2):1039--1068, 2020.

\bibitem[\protect\citeauthoryear{Pruhs and Woeginger}{2012}]{PW12divorcing}
Kirk Pruhs and Gerhard~J Woeginger.
\newblock {Divorcing Made Easy}.
\newblock In {\em International Conference on Fun with Algorithms}, pages
  305--314, 2012.

\bibitem[\protect\citeauthoryear{Saban and Sethuraman}{2014}]{saban2014note}
Daniela Saban and Jay Sethuraman.
\newblock {A Note on Object Allocation under Lexicographic Preferences}.
\newblock {\em Journal of Mathematical Economics}, 50:283--289, 2014.

\bibitem[\protect\citeauthoryear{Schmitt and Martignon}{2006}]{SM06complexity}
Michael Schmitt and Laura Martignon.
\newblock {On the Complexity of Learning Lexicographic Strategies}.
\newblock {\em Journal of Machine Learning Research}, 7(Jan):55--83, 2006.

\bibitem[\protect\citeauthoryear{Schulman and Vazirani}{2015}]{SV15allocation}
Leonard~J Schulman and Vijay~V Vazirani.
\newblock {Allocation of Divisible Goods Under Lexicographic Preferences}.
\newblock In {\em 35th IARCS Annual Conference on Foundations of Software
  Technology and Theoretical Computer Science}, page 543, 2015.

\bibitem[\protect\citeauthoryear{Sikdar \bgroup \em et al.\egroup
  }{2017}]{Sikdar2017:Mechanism}
Sujoy Sikdar, Sibel Adali, and Lirong Xia.
\newblock {Mechanism Design for Multi-Type Housing Markets}.
\newblock In {\em Thirty-First AAAI Conference on Artificial Intelligence},
  pages 684--690, 2017.

\bibitem[\protect\citeauthoryear{Taylor}{1970}]{T70problem}
Michael Taylor.
\newblock {The Problem of Salience in the Theory of Collective
  Decision-Making}.
\newblock {\em Behavioral Science}, 15(5):415--430, 1970.

\bibitem[\protect\citeauthoryear{Varian}{1974}]{varian1974equity}
Hal~R. Varian.
\newblock {Equity, Envy, and Efficiency}.
\newblock {\em Journal of Economic Theory}, 9(1):63--91, 1974.

\bibitem[\protect\citeauthoryear{Zhou}{1990}]{zhou1990conjecture}
Lin Zhou.
\newblock {On a Conjecture by Gale about One-Sided Matching Problems}.
\newblock {\em Journal of Economic Theory}, 52(1):123--135, 1990.

\end{thebibliography}

\clearpage
\begin{center}
\Large{Appendix}
\end{center}

\section{Proof of Proposition~\ref{prop:minimalityGoods}}
\label{sec:app:prop:minimalityGoods}

Recall from \Cref{thm:EFX_PO_SP_Neutral_NonBossy} that the set of all mechanisms satisfying \EFX{}, \PO{}, strategyproofness, non-bossiness, and neutrality is characterized by the family of mechanisms
in Algorithm~\ref{alg:IQSD} (parameterized by $\sigma$). In \Cref{prop:minimalityGoods}, we will show that the aforementioned set of properties is \emph{minimal} in the sense that dropping any one of them necessarily expands the set of mechanisms beyond those in Algorithm~\ref{alg:IQSD}.

\minimalityGoods*

\begin{proof}
\textbf{\EFX{} is necessary}:
    Consider any serial dictatorship quota (\SDQ{}) mechanism where the first agent has a quota of $2$. From \Cref{prop:quota}, we know that this mechanism is \PO{}, strategyproof, non-bossy, and neutral. However, when agents have identical preferences, this mechanism fails to be \EFX{} since the first agent, who is envied by everyone else, receives more than one good (\Cref{prop:efx_property}).

\textbf{\PO{} is necessary}: A mechanism that leaves everything unassigned is \EFX{}, strategyproof, non-bossy, and neutral, but clearly violates Pareto optimality.
    
\textbf{Non-bossiness is necessary}: Consider the following ``conditional'' variant of Algorithm~\ref{alg:IQSD} with four agents and four goods: If the first two agents, $a_1$ and $a_2$, have identical preferences, then the ordering of agents is $\sigma=(a_1,a_2,a_3,a_4)$; otherwise the priority ordering is $\sigma'=(a_1,a_2,a_4,a_3)$.

    It is easy to see that this mechanism is \EFX{} (since each agent gets exactly one item), \PO{} (because of sequencibility), and neutral (because each agent's bundle is a singleton). It is also not hard to see that the mechanism is strategyproof: Indeed, agents $a_1$ and $a_2$ occur in the first and the second positions, respectively, under both orderings, and therefore cannot benefit by misreporting their preferences. Further, agents $a_3$ and $a_4$ cannot alter which ordering is invoked by misreporting, and for any fixed ordering, have no incentive to misreport either (since each instantiation of the mechanism is a serial dictatorship, which is strategyproof).
    
    To see why this mechanism violates non-bossiness, consider the following preference profile:
    \begin{align*}
    a_1: g_1 \, \> \, g_2 \, \> \, g_3 \, \> \, g_4\\ \nonumber
    a_2: g_1 \, \> \, g_2 \, \> \, g_3 \, \> \, g_4\\ \nonumber
    a_3: g_3 \, \> \, g_4 \, \> \, g_1 \, \> \, g_2\\ \nonumber
    a_4: g_3 \, \> \, g_4 \, \> \, g_1 \, \> \, g_2 \nonumber
    \end{align*}
    Here, agent $a_2$ can change the allocation of agent $a_3$ by misreporting her preferences as
    $$a_2: g_2 \, \> \, g_1 \, \> \, g_3 \, \> g_4$$ without changing her own allocation.

\textbf{Neutrality is necessary}: Consider the following ``conditional'' variant of Algorithm~\ref{alg:IQSD}, where agent $1$ always picks a good first, and the ordering over the rest of the agents is decided as follows: if agent $1$ picks $g_1$, then $\sigma=(1,2,3,\dots,n)$. Otherwise, the ordering is $\sigma'=(1,n, n-1,\dots,2)$. It is easy to see that this mechanism is \EFX{}, \PO{}, strategyproof, and non-bossy, but not neutral.

\textbf{Strategyproofness is necessary}: Consider the following mechanism variant of Algorithm~\ref{alg:IQSD}: (i) Fix a priority ordering, say $\sigma=(1,\cdots,n)$, and execute one round of serial dictatorship according to $\sigma$. (ii) Assign the remaining goods according to the following rule: If agent $n$ envies agent $n-1$, all remaining goods are assigned to $n$. Otherwise, agents $n-1$ and $n$ pick the remaining goods in a round-robin fashion.
    
    This mechanism is not strategyproof because in step (ii) the round-robin picking order is manipulable. In particular, agent $n$ can obtain \emph{all} the goods that remain after one round of serial dictatorship by misreporting her preferences (and pretending to envy agent $n-1$).
    
    It is easy to verify that the mechanism is \EFX{}, \PO{}, and neutral. We will now argue that this mechanism is also non-bossy. The intuitive idea is as follows: The first $n-2$ agents cannot change the outcome of any other agent by misreporting without changing their own allocation. Subject to the items allocated to these agents, the bundles of agents $n-1$ and $n$ are complements of each other, and thus, changing one implies changing the other.
    
    Formally, let us suppose, for the sake of contradiction, that the mechanism is bossy. There are two possible cases.
    
    \textbf{Case 1}: Suppose agent $n$ envies $n-1$ under truthful reporting. (i) If $n -1$ misreports to avoid envy from agent $n$, then agent $n$ will pick the good that was assigned to agent $n-1$ under truthful reporting, and thus agent $n -1$'s outcome must change. (ii) If $n$ misreports so to avoid envy, then $n$'s outcome must change since it will now receive half of the remaining goods.
    
    \textbf{Case 2}: Suppose agent $n$ does not envy agent $n -1$ under truthful reporting. (i) If $n$ reports his preferences to declare envy towards $n -1$, then his allocation must change since now he receives all the remaining goods, instead of receiving only half of the remaining goods under truthful reporting. (ii) Agent $n-1$ can only declare envy by picking a less desired good before agent $n$, and thus, changing its outcome, which contradicts bossiness.
\end{proof}

\section{Proof of Theorem~\ref{thm:EFk_RM_NP-complete}}
\label{sec:app:thm:EFk_RM_NP-complete}

For any fixed $k \in \mathbb{N}$, the computational problem \EFkRMExistence{k} asks whether a given instance admits an \EF{k} and rank-maximal allocation (\Cref{defn:EFkRMExistence}). Note that the parameter $k$ is fixed in advance and is not a part of the input. 

\begin{dfn}[\EFkRMExistence{k}]
Given any fair division instance with lexicographic preferences, does there exist an allocation that is envy-free up to $k$ goods $(\EF{k})$ and rank-maximal $(\RM)$?
\label{defn:EFkRMExistence}
\end{dfn}

\EFkRM*
\begin{proof}
Given an allocation, one can check in polynomial time whether it satisfies \EF{k} and rank-maximality. Hence, \EFkRMExistence{k} is in \NP{}.

To prove \NPH{}ness, we will show a reduction from \PITFull{} (\PIT{}), which is known to be \NPC{}~\citep[Page 68]{Garey79:Computers}. An instance of \PIT{} consists of a graph $G = (V,E)$ with $|V|=3q$ vertices for some $q \in \mathbb{N}$. The goal is to determine whether there exists a partition of the vertex set $V$ into $q$ disjoint sets $V^1,\dots,V^q$ of three vertices each such that for every $V^i = \{v_{i,1},v_{i,2},v_{i,3}\}$, the three edges $\{v_{i,1},v_{i,2}\}$, $\{v_{i,2},v_{i,3}\}$, and $\{v_{i,3},v_{i,1}\}$ all belong to $E$.

 We will assume that $G$ is a \emph{balanced tripartite} graph, that is, the set of vertices $V$ can be partitioned into three disjoint sets $W$, $X$, and $Y$ such that $|W| = |X| = |Y| = q$ and no pair of vertices within the same set are adjacent in $E$. It is known that \PIT{} remains \NPC{} even under this restriction~\citep[Proposition 5.1]{CKW15geometric}. We will write $W = \{w_1,\dots,w_q\}$, $X = \{x_1,\dots,x_q\}$, and $Y = \{y_1,\dots,y_q\}$ to denote the vertices of the graph $G$. In addition, we will use $t \coloneqq \binom{3q}{2} - |E|$ to denote the number of \emph{non-edges} in $G$.

\emph{Construction of the reduced instance}: We will construct a fair division instance with $n$ agents and $m$ goods, where $n = q + t \left( (k+2)q + 1 \right)$ and $m = 3q + (k-1)q + t \left( (k+2)q + 1 \right)$. 

The set of agents consists of $q$ \emph{main} agents $a_1,\dots,a_{q}$, and $t$ groups of \emph{dummy} agents, each group comprising of $(k+2)q + 1$ agents. The dummy agents in the $\ell^\text{th}$ group are denoted by $d^\ell_1,d^\ell_2,\dots,d^\ell_{(k+2)q + 1}$.

The set of goods consists of $3q$ \emph{main} goods $\{W_1,\dots,W_q,X_1,\dots,X_q,Y_1,\dots,Y_q\}$, $(k-1)q$ \emph{selector} goods and $t \left( (k+2)q + 1 \right)$ \emph{dummy} goods. The selector goods are classified into $q$ groups of $k-1$ goods each, where the goods in the $\ell^\text{th}$ group are denoted by $S^\ell_1,S^\ell_2,\dots,S^\ell_{k-1}$. The dummy goods are classified into $t$ groups of $(k+2)q + 1$ goods each, where the goods in the $\ell^\text{th}$ group are denoted by $D^\ell_1,D^\ell_2,\dots,D^\ell_{(k+2)q + 1}$.

For any $i,j \in \mathbb{N}$ such that $i \leq j$, we will use the compact notation $D^\ell_{i:j}$ to denote the set $\{D^\ell_{i}, D^\ell_{i+1}, \dots, D^\ell_{j-1}, D^\ell_{j}\}$. Similarly, we will write $S^\ell_{i:j}$, $W_{i:j}$, $X_{i:j}$, and $Y_{i:j}$ to denote $\{S^\ell_{i}, S^\ell_{i+1}, \dots, S^\ell_{j}\}$, $\{W_{i}, W_{i+1}, \dots, W_{j}\}$, $\{X_{i}, X_{i+1}, \dots, X_{j}\}$, and $\{Y_{i}, Y_{i+1}, \dots, Y_{j}\}$, respectively.

We will now describe the preferences of the agents. Let $\vartriangleright$ be a total ordering on the set of all goods given by
\begin{align*}
\vartriangleright \, : \, & S^1_{1:k-1} \> \, S^2_{1:k-1} \> \, \dots \> \, S^q_{1:k-1} \> \, W_{1:q} \> \, X_{1:q} \> \, Y_{1:q} \> \, D^1_{1:(k+2)q+1} \> \, D^2_{1:(k+2)q+1} \> \, \dots \> \, D^t_{1:(k+2)q+1},
\end{align*}
where the goods within each group are ordered in $\vartriangleright$ according to their indices. Thus, for instance, the ordering within the group $S^1_{1:k-1}$ is $S^1_1 \> \, S^1_2 \> \, \dots \> \, S^1_{k-1}$.

For every $i \in [q]$, the preferences of main agent $a_i$ are given by
\begin{align*}
a_i: S^i_{1:k-1} \> \, & S^{i+1}_{1:k-1} \> \, \dots \> \, S^q_{1:k-1} \> \, S^1_{1:k-1} \> \, S^2_{1:k-1} \> \, \dots \> \, S^{i-1}_{1:k-1} \> \, W_{1:q} \> \, X_{1:q} \> \, Y_{1:q} \> \, *,
\end{align*}
where $*$ denotes the rest of the goods ordered according to $\vartriangleright$. Notice that the groups of selector goods are ordered in a cyclic manner by the main agents. In particular, the main agent $a_i$ prefers the selector group $S^i_{1:k-1}$ over any other set of goods.

In order to specify the preferences of the dummy agents, recall that there are as many groups of dummy agents as there are non-edges in $G$. Let the dummy agents in the $\ell^\text{th}$ group, namely $\{d^\ell_1,\dots,d^\ell_{(k+2)q+1}\}$, be associated with the non-edge $\{v_i,v_j\} \notin E$, where $v_i,v_j \in W \cup X \cup Y$. All dummy agents in the $\ell^\text{th}$ group have identical preferences: For any $r \in \{1,\dots,(k+2)q+1\}$, the preferences of the dummy agent $d^\ell_r$ are given by
\begin{align*}
d^\ell_r: D^\ell_{1:(k+2)q} \> \, V_i \> \, V_j \> \, & S^1_{1:k-1} \> \, \dots \> \, S^q_{1:k-1} \> \, D^\ell_{(k+2)q+1} \> \, *,
\end{align*}
where $V_i$ and $V_j$ are the main goods associated with the vertices $v_i$ and $v_j$, respectively, and $*$ once again denotes the rest of the goods ordered according to $\vartriangleright$. This completes the construction of the reduced instance.

Observe that all main goods are ranked inside the top $(k+2)q$ positions by each main agent, whereas each dummy agent ranks any main good outside the top $(k+2)q$. As a result, any rank-maximal allocation for the above instance must assign the main goods exclusively among the main agents. By similar reasoning, for every $i \in [q]$, the selector goods in the $i^\text{th}$ group, namely $S^i_{1:k-1}$, are assigned to the main agent $a_i$, and for every $\ell \in [t]$, all dummy goods in the $\ell^\text{th}$ group, namely $D^\ell_{1:(k+2)q+1}$, are assigned exclusively among the dummy agents in the $\ell^\text{th}$ group $\{d^\ell_1,\dots,d^\ell_{(k+2)q+1}\}$. This necessary condition is also sufficient for rank-maximality, since the main agents have identical ranking of the main goods and the dummy agents in the $\ell^\text{th}$ group have identical ranking of the dummy goods in the $\ell^\text{th}$ group.

We will now argue the equivalence of the solutions.

$(\Rightarrow)$ Let $V^1 \sqcup V^2 \sqcup \dots \sqcup V^q$ be a solution of \PIT{} instance (where $\sqcup$ denotes disjoint union of sets). The desired allocation $A$ can be constructed as follows: For every $i \in [q]$, assign the main goods corresponding to the vertices in $V^i$ to the main agent $a_i$. Next, for every $i \in [q]$, assign the selector goods in the $i^\text{th}$ group, namely $S^i_{1:k-1}$, to the main agent $a_i$. Finally, for every $\ell \in [t]$ and every $i \in \{1,\dots,(k+2)q+1\}$, assign the dummy good $D^\ell_i$ to the dummy agent $d^\ell_i$.

Note that the allocation $A$ assigns each good to exactly one agent; in particular, each main good is assigned to exactly one main agent since the sets $V^1, \dots, V^q$ constitute a partition of the vertex set $V$. Additionally, $A$ satisfies the aforementioned sufficient condition for rank-maximality. All that remains to be shown is that $A$ is \EF{k}.

Since each dummy agent $d^\ell_i$ gets exactly one good, the envy of any other agent towards $d^\ell_i$ can be eliminated by the removal of this good. Therefore, in order to establish that $A$ satisfies \EF{k}, it suffices to bound the envy towards the main agents. Furthermore, each main agent $a_i$ receives its favorite good, namely $S^i_1$, under $A$, and therefore does not envy any other agent (because of lexicographic preferences). Thus, it suffices to focus only on the envy experienced by the dummy agents towards the main agents.

Fix a main agent $a_r$. Suppose, for some $\ell \in [t]$, some dummy agent in the $\ell^\text{th}$ group envies $a_r$. Then, without loss of generality, the dummy agent $d^\ell_{(k+2)q+1}$ must also envy $a_r$. This is because all dummy agents in the same group have identical preferences and the $d^\ell_{(k+2)q+1}$ receives the worst bundle among the agents in its group.

Recall that each group of dummy agents is associated with a non-edge. Let $(v_i,v_j) \notin E$ be the non-edge associated with the $\ell^\text{th}$ dummy group. For a violation of \EF{k} to occur between $a_r$ and $d^\ell_{(k+2)q+1}$, the agent $a_r$ should receive $k+1$ or more goods that are preferred by $d^\ell_{(k+2)q+1}$ over its favorite good in its bundle, namely $D^\ell_{(k+2)q+1}$. This is possible only if, in addition to $k-1$ selector goods, $a_r$ also receives the main goods $V_i$ and $V_j$ in the allocation $A$.

Recall that the main goods assigned to the agent $a_r$ come from the set $V^r$, which corresponds to a `triangle' in the graph $G$. Thus, each pair of main goods assigned to a main agent must correspond to vertices that constitute an edge in the graph $G$. This, however, contradicts the fact that $(v_i,v_j)$ is a non-edge. Therefore, the allocation $A$ must satisfy \EF{k}, as desired.

($\Leftarrow$) Now suppose that there exists an allocation $A$ that is \EF{k} and rank-maximal. Then, $A$ must satisfy the aforementioned necessary condition for rank-maximality. That is, all main goods are assigned among the main agents, the selector goods in the $i^\text{th}$ group are assigned to the main agent $a_i$, and all dummy goods in the $\ell^\text{th}$ group are assigned among the dummy agents in the $\ell^\text{th}$ group.

We will now argue that each dummy agent must get at least one good in $A$. Indeed, rank-maximality requires that each main agent is assigned $k-1$ selector goods. In addition, since there are $3q$ main goods to be allocated among $q$ main agents, some main agent, say $a_r$, must get at least three main goods, resulting in a total of at least $k+2$ goods in its bundle. If some dummy agent, say $d^\ell_i$, gets an empty bundle, then $d^\ell_i$ will envy $a_r$, and the envy cannot be eliminated by the removal of $k$ goods, creating a violation of \EF{k}. Thus, every dummy agent must get at least one good. Furthermore, for any $\ell \in [t]$, there are $(k+2)q+1$ dummy agents in the $\ell^\text{th}$ group and equally many dummy goods. Thus, each dummy agent gets exactly one dummy good in $A$. Without loss of generality, we can assume that the agent $d^\ell_i$ gets the good $D^\ell_i$. In particular, the only good assigned to the agent $d^\ell_{(k+2)q+1}$ is $D^\ell_{(k+2)q+1}$.

We will now argue that each main agent gets exactly three main goods. By way of contradiction, suppose that some main agent, say $a_r$, gets four or more main goods. Then, at least two of these goods, say $V_i$ and $V_j$, must both be $W$-goods (or $X$-goods or $Y$-goods). This means that the corresponding vertices $v_i$ and $v_j$ must both belong to either $W$ or $X$ or $Y$, implying that $\{v_i,v_j\}$ is a non-edge. By construction, each non-edge is associated with a group of dummy agents, so let the non-edge $\{v_i,v_j\}$ be associated with the $\ell^\text{th}$ dummy group. This would imply that $d^{\ell}_{(k+2)q+1}$ prefers $k+1$ goods assigned to $a_r$---namely, the two main goods $V_i$ and $V_j$ as well as $k-1$ selector goods $S^r_{1:k-1}$---over the only good in its own bundle, namely $D^{\ell}_{(k+2)q+1}$. Therefore, the envy of $d^{\ell}_{(k+2)q+1}$ towards $a_r$ cannot be addressed by removing at most $k$ goods from $A_r$. This indicates a violation of \EF{k}, which is a contradiction. Therefore, each main agent must get exactly three main goods. The above argument also establishes that these three main goods must come from three different sets $W$, $X$, and $Y$. Thus, for every $r \in [q]$, the bundle assigned to agent $a_r$ is of the form $A_r = \{W_h,X_i,Y_j,S^r_{1:k-1}\}$.

We will now show that one can infer a solution of \PIT{} from the allocation $A$. Indeed, for every $r \in [q]$, let $V^r \coloneqq \{w_h,x_i,y_j\}$ whenever $\{W_h,X_i,Y_j\} \in A_r$. Notice that the sets $V^1,\dots,V^q$ constitute a valid partition of the vertex set of $G$. This is because each main good is assigned to exactly one main agent in $A$ and therefore each of the corresponding vertices is assigned to exactly one of the sets $V^1,\dots,V^q$. Furthermore, each set $V^r = \{v_{r,1},v_{r,2},v_{r,3}\}$ is a `triangle', i.e., each pair of vertices in $V^r$ must form an edge in $G$, i.e., $\{v_{r,1},v_{r,2}\}, \{v_{r,2},v_{r,3}\}, \{v_{r,1},v_{r,3}\} \in E$. Indeed, if some pair of vertices in $V^r$ is a non-edge, then by earlier reasoning, some dummy agent will create an \EF{k} violation with the main agent $a_r$. Therefore, the sets $V^1,\dots,V^q$ constitute a valid solution of \PIT{}.

This completes the proof of \Cref{thm:EFk_RM_NP-complete}.
\end{proof}

\section{Proof of Proposition~\ref{prop:EF1_RM_ThreeAgents}}
\label{sec:app:prop:EF1_RM_ThreeAgents}

Recall the statement of \Cref{prop:EF1_RM_ThreeAgents}:
\EFoneRMThreeAgents*

\begin{proof}
If the top-ranked goods of the three agents are all distinct, then an envy-free and rank-maximal allocation exists (\Cref{prop:EF_RM_Polytime}). Otherwise, suppose exactly two agents, say $a_1$ and $a_2$, have the same top-ranked good (the case where the top-ranked goods of all three agents coincide can be handled similarly).

We will start by carrying out all assignments that are uniquely determined by the rank-maximality condition. That is, if there is a unique agent for whom the highest ranking of a fixed good $g$ is realized, then we assign $g$ to that agent. 

Next, let us consider two cases based on who out of $a_1$ or $a_2$ gets its top-ranked good, say $g'$. If $a_1$ gets $g'$, then $a_2$ is given the highest-ranked good in its list that it could be assigned without violating rank-maximality. If the corresponding partial allocation is \EF{1}, then we can extend it to a rank-maximal allocation. Otherwise, we consider the other case where $a_2$ gets $g'$ and $a_1$ gets the highest-ranked feasible good. If no \EF{1} partial allocation exists, the algorithm returns NO.
\end{proof}

\section{Comparing \MMS{} with Relaxations of Envy-freeness}
\label{sec:app:prop:EFX_implies_MMS}

Let us start by proving that \EFX{} is a strictly stronger notion than \MMS{} on the domain of lexicographic preferences.

\begin{restatable}[]{prop}{EFXimpliesMMS}
Given any instance with lexicographic preferences, any \EFX{} allocation satisfies maximin share guarantee $(\MMS{})$ but the converse is not always true.
\label{prop:EFX_implies_MMS}
\end{restatable}
\begin{proof}
Suppose, for contradiction, that there is an \EFX{} allocation $A$ that is not \MMS{}. Then, there must exist some agent, say $i$, that receives a strict subset of its bottom-$(m-n+1)$ goods. (This is because under lexicographic preferences, if an agent receives one or more of its top-$(n-1)$ goods, or if it receives all of its bottom-$(m-n+1)$ goods, then its maximin share guarantee is satisfied.)

Let $S$ denote the set of top-$(n-1)$ goods according to agent $i$'s preference. By the above observation, the goods in $S$ are allocated among the other $n-1$ agents. In order for the allocation $A$ to be \EFX{}, no agent should get more than one good in $S$ (since we know from \Cref{prop:efx_property} that any envied agent gets exactly one good in an \EFX{} allocation). Thus, each agent gets exactly one good in $S$, and therefore, agent $i$ envies every other agent.

Since agent $i$'s bundle is a \emph{strict} subset of its bottom-$(m-n+1)$ goods, there must exist a good in $M \setminus S$ that is assigned to an agent other than $i$. This, however, results in an envied agent getting more than one good, thus violating the assumption that $A$ is \EFX{}~(\Cref{prop:efx_property}). Therefore, $A$ must satisfy \MMS{}.

To prove that \MMS{} does not always imply \EFX{}, consider the following instance with four agents and five goods:
\begin{align*}
    a_1: g_1 \, \> \, g_2 \, \> \, g_3 \, \> \, g_4 \, \> \, g_5\\ \nonumber
    a_2: g_1 \, \> \, g_2 \, \> \, g_3 \, \> \, g_4 \, \> \, g_5\\ \nonumber
    a_3: g_4 \, \> \, g_1 \, \> \, g_2 \, \> \, g_3 \, \> \, g_5\\ \nonumber
    a_4: g_5 \, \> \, g_1 \, \> \, g_2 \, \> \, g_3 \, \> \, g_4 \nonumber
\end{align*}
The allocation $\{\{g_1,g_2\},\{g_3\},\{g_4\},\{g_5\}\}$ is \MMS{} since each agent receives one or more of its top-$(n-1)$ goods. However, it is not \EFX{} because agent $a_1$, who is envied by $a_2$, receives more than one good~(\Cref{prop:efx_property}).
\end{proof}

It can also be shown that \EF{1} and \MMS{} are incomparable notions in that one does not always imply the other. The fact that \MMS{} does not imply \EF{1} follows from the example in the proof of \Cref{prop:EFX_implies_MMS}: Indeed, agent $a_2$ continues to envy agent $a_1$ even after the removal of any good from the latter's bundle. To prove that \EF{1} does not imply \MMS{}, consider the following instance with identical preferences:
\begin{align*}
    a_1: g_1 \, \> \, g_2 \, \> \, g_3 \, \> \, g_4 \, \> \, g_5\\ \nonumber
    a_2: g_1 \, \> \, g_2 \, \> \, g_3 \, \> \, g_4 \, \> \, g_5\\ \nonumber
    a_3: g_1 \, \> \, g_2 \, \> \, g_3 \, \> \, g_4 \, \> \, g_5\\ \nonumber
    a_4: g_1 \, \> \, g_2 \, \> \, g_3 \, \> \, g_4 \, \> \, g_5 \nonumber
\end{align*}
The allocation $\{\{g_4\},\{g_1,g_5\},\{g_2\},\{g_3\}\}$ is \EF{1}; in particular, any agent's envy towards $a_2$ can be eliminated by the removal of the good $g_1$. However, it fails \MMS{} since $a_1$ gets a strict subset of its bottom-$(m-n+1)$ goods.

\section{MMS Characterization and Minimality}
\label{sec:app:mms}

When requiring  \PO{}, strategyproofness, non-bossiness, and neutrality, mechanisms satisfying \MMS{} are restricted to those presented in Algorithm~\ref{alg:IQSD}. This is due the fact that, subject to these additional properties, a mechanism satisfies \EFX{} if and only if it satisfies \MMS{}.

\begin{restatable}{prop}{MMSPropertyGoods}
An allocation of goods is \MMS{} if and only if each agent's bundle consists of one or more goods from among its top-$(n-1)$ goods or all of its bottom $m-n+1$ goods.
\label{prop:mms_property_Goods}
\end{restatable}
\begin{proof}
Consider the agent $i \in N$. Let $\succ_i = g_1 \succ g_2 \succ \dots \succ g_m$. Then, the \MMS{} partition of agent $i$ is uniquely defined, and its \MMS{} value is given by $\MMS_i = \min \left\{\{c_1\},\{c_2\},\dots, \{c_{n}, \ldots, c_m\}\right\}$, where $\min\{\cdot\}$ denotes the least-preferred bundle with respect to the lexicographic extension of $\succ_i$.
\end{proof}

Therefore, we have the following characterization.

\begin{restatable}{thm}{thmMMSCh}
For any ordering $\sigma$ of the agents, Algorithm~\ref{alg:IQSD} is \MMS{}, \PO{}, strategyproof, non-bossy, and neutral. Conversely, any mechanism satisfying these properties can be implemented by Algorithm~\ref{alg:IQSD} for some $\sigma$.
\label{thm:EFX_PO_SP_Neutral_NonBossy_MMS}
\end{restatable}

Notice that the \EFX{} and \MMS{} allocations do not necessarily coincide, but due to~\Cref{prop:quota}, under strategyproofness, non-bossiness, and neutrality, the set of mechanisms is restricted to \SDQ{}s, guaranteeing that \EFX{} and \MMS{} allocations do coincide. Nonetheless, although the family of mechanisms satisfying \MMS{} or \EFX{} along with \PO{}, strategyproofness, non-bossiness, and neutrality are the same, the proof of minimality of these properties include subtle differences.

\begin{restatable}[]{thm}{MMSMinimality}
The set of \MMS{}, \PO{}, strategyproofness, non-bossiness, and neutrality is a minimal set of properties for characterizing Algorithm~\ref{alg:IQSD}.
\label{thm:MMS_PO_SP_Neutral_NonBossy}
\end{restatable}

\begin{proof}
The proofs of necessity of MMS, non-bossiness, and neutrality are identical to those in \Cref{prop:minimalityGoods} by replacing EFX with MMS. 

\noindent \textbf{PO is necessary}.
Consider the following conditional mechanism: (i) Fix a priority ordering, say $\sigma=(1,\cdots,n)$, and execute one round of serial dictatorship according to $\sigma$.
(ii) Assign the remaining goods according to the following rule: 
if agent $n$ receives a good that is ranked $n$ according to $\succ_{n}$, then all the remaining goods are assigned to $n$.
Otherwise, throw away the remaining goods.

By \cref{prop:mms_property_Goods}, this mechanism is MMS. This rule is an SDQ mechanism with a quota of $(1, \ldots, m-n+1)$ or $(1,\ldots, 1)$. Thus, it is neutral, non-bossy, and strategyproof. However, the mechanism is not PO since it may throw away some goods if the quota of $(1,\ldots, 1)$ is selected.

\noindent \textbf{Strategyproofness is necessary}.
Consider the following mechanism as a variant of Algorithm~\ref{alg:IQSD}: (i) Fix a priority ordering, say $\sigma=(1,\cdots,n)$, and execute one round of serial dictatorship according to $\sigma$.
(ii) Assign the remaining goods according to the following rule: 
if agent $n$ receives a good that is ranked $n$ according to $\succ_{n}$, then all the remaining goods are assigned to $n$. Otherwise, agents $n-1$ and $n$ pick the remaining goods in a round-robin fashion.

Notice that the above mechanism is MMS according to \Cref{prop:mms_property_Goods}. The rest of the proof follows precisely as the one for \Cref{prop:minimalityGoods}.
\end{proof}

\section{Experiments}\label{sec:app:experiments}

In \Cref{fig:mallows}, we observe that the difference between the fraction of instances that admit an \EFX{}+\RM{} allocation and the fraction of instances that admit an \EF{}+\RM{} allocation follows a decreasing trend as the number of goods increase. We conjecture that this is because of the following reason: From~\Cref{prop:efx_property}, we know that under an \EFX{} allocation, no envied agent can be assigned more than one good. Thus, in a regime where an EF allocation doesn't exist, a sufficient condition for the non-existence of an EFX+RM allocation is that each agent is assigned two or more goods. 

To test this conjecture experimentally, we find the fraction of ``good'' instances that admit an \RM{} allocation where some agent receives at most one good, using the data from our experiments in \Cref{fig:mallows}. We limit our experiments to instances with $m=50$ goods, because we observed that for larger values of $m$, the integer linear programs to check whether an instance is ``good'' or admits an \EF{} allocation take too long to run. We omit the case of $\phi=0$ which induces profiles with identical preferences, where an \EF{}+\RM{} allocation never exists, and an \EFX{}+\RM{} allocation always exists.

\begin{figure}[h]
	\centering
	\includegraphics[width=0.49\linewidth]{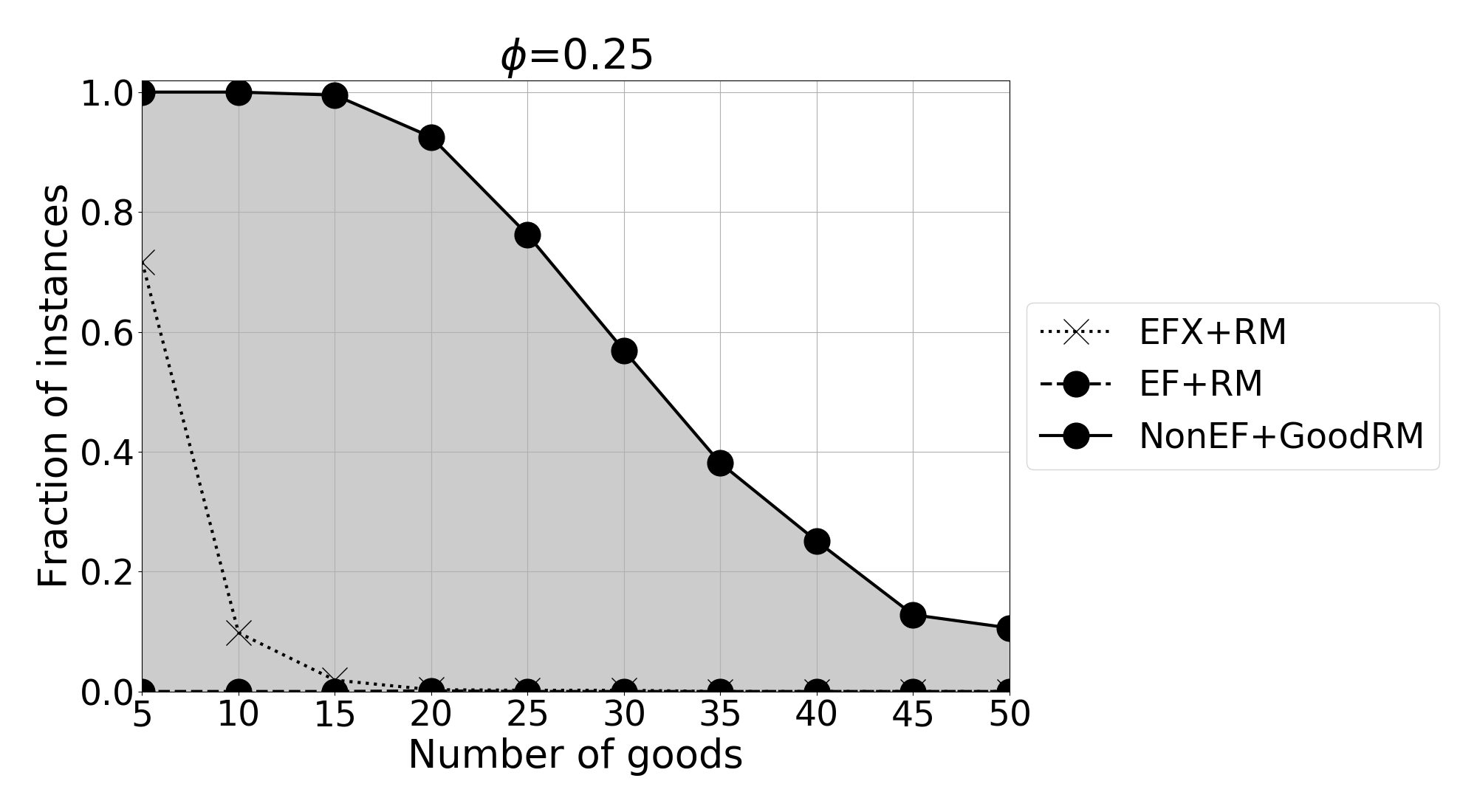}
	\hfill \includegraphics[width=0.49\linewidth]{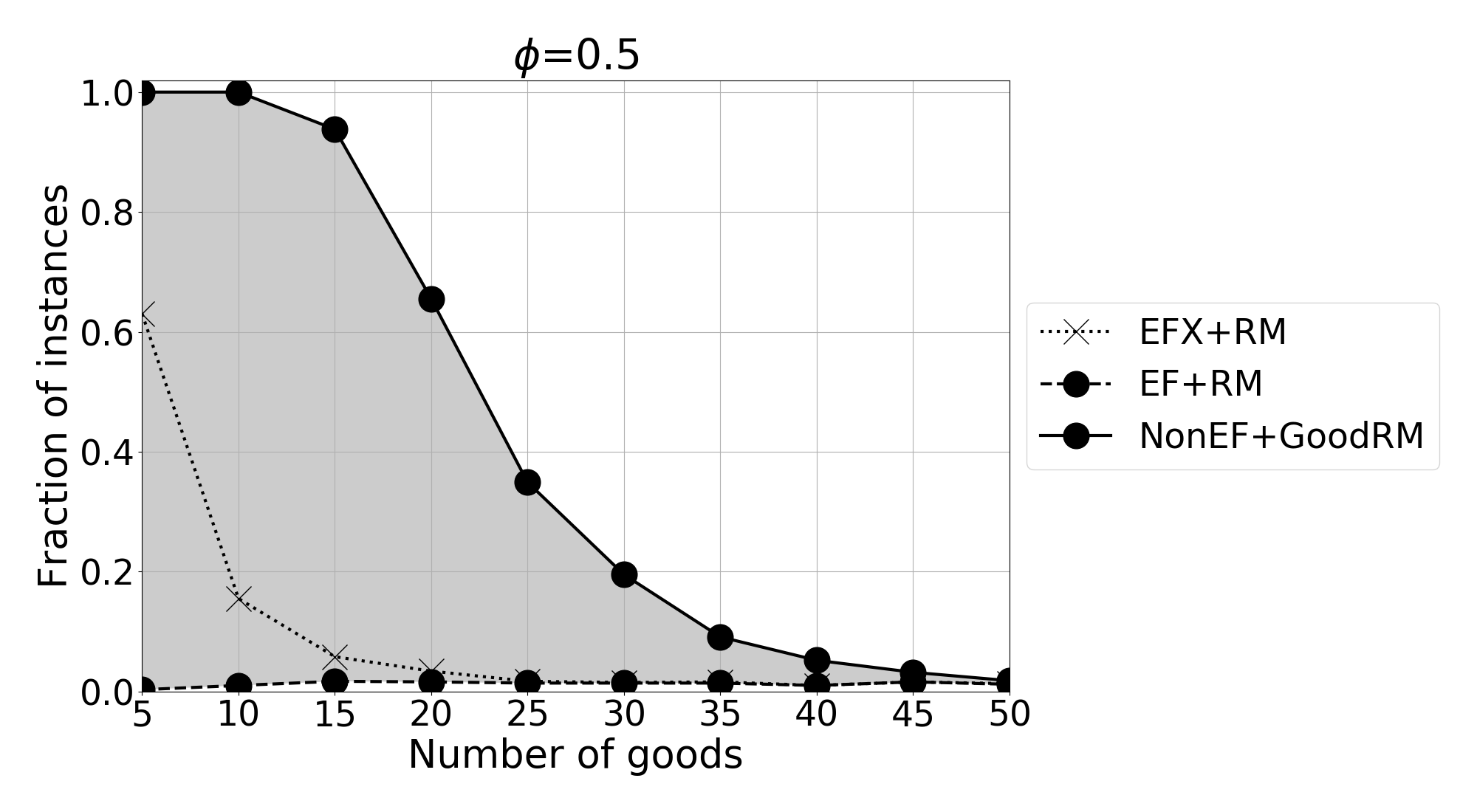}\\
	\includegraphics[width=0.49\linewidth]{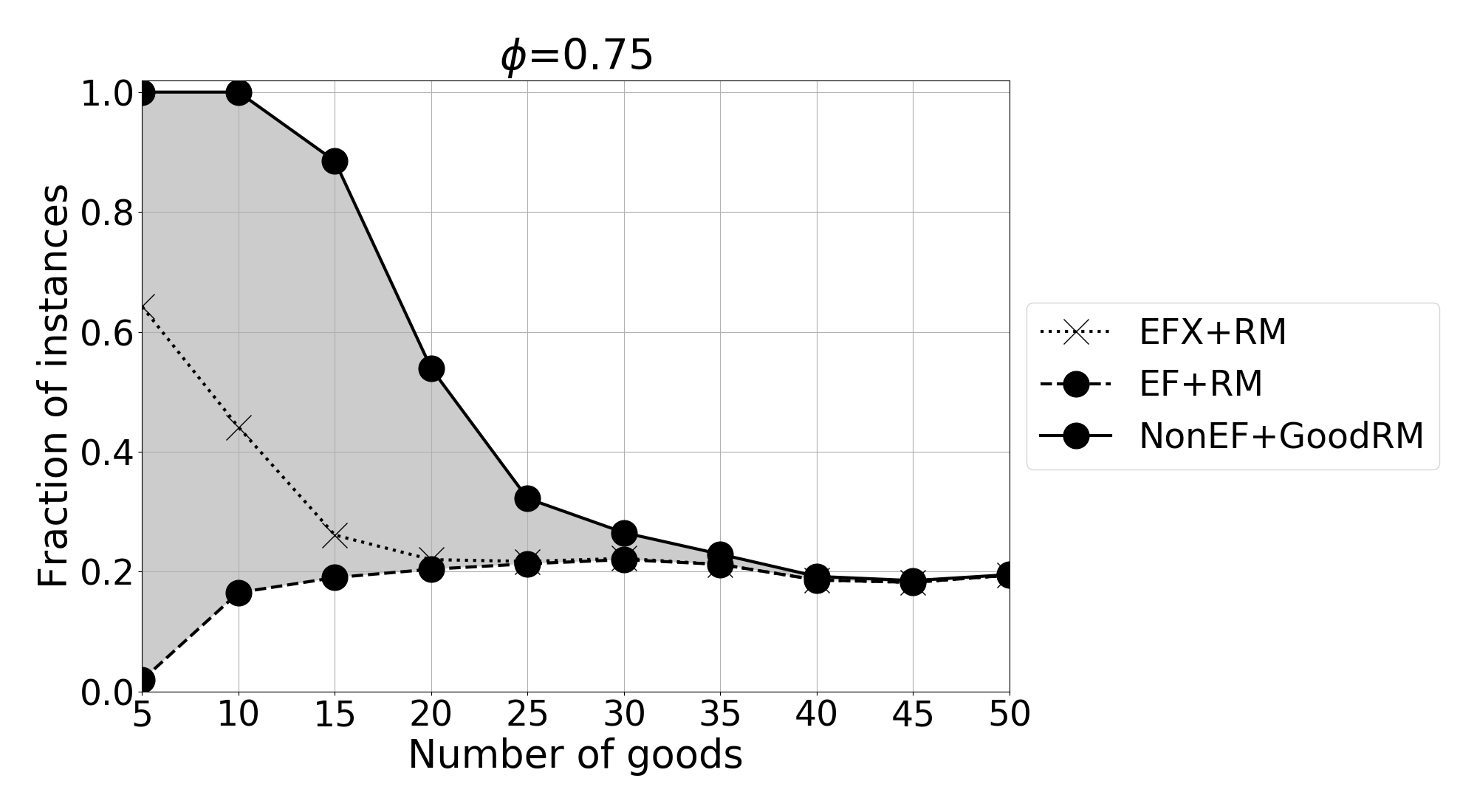}
	\hfill	\includegraphics[width=0.49\linewidth]{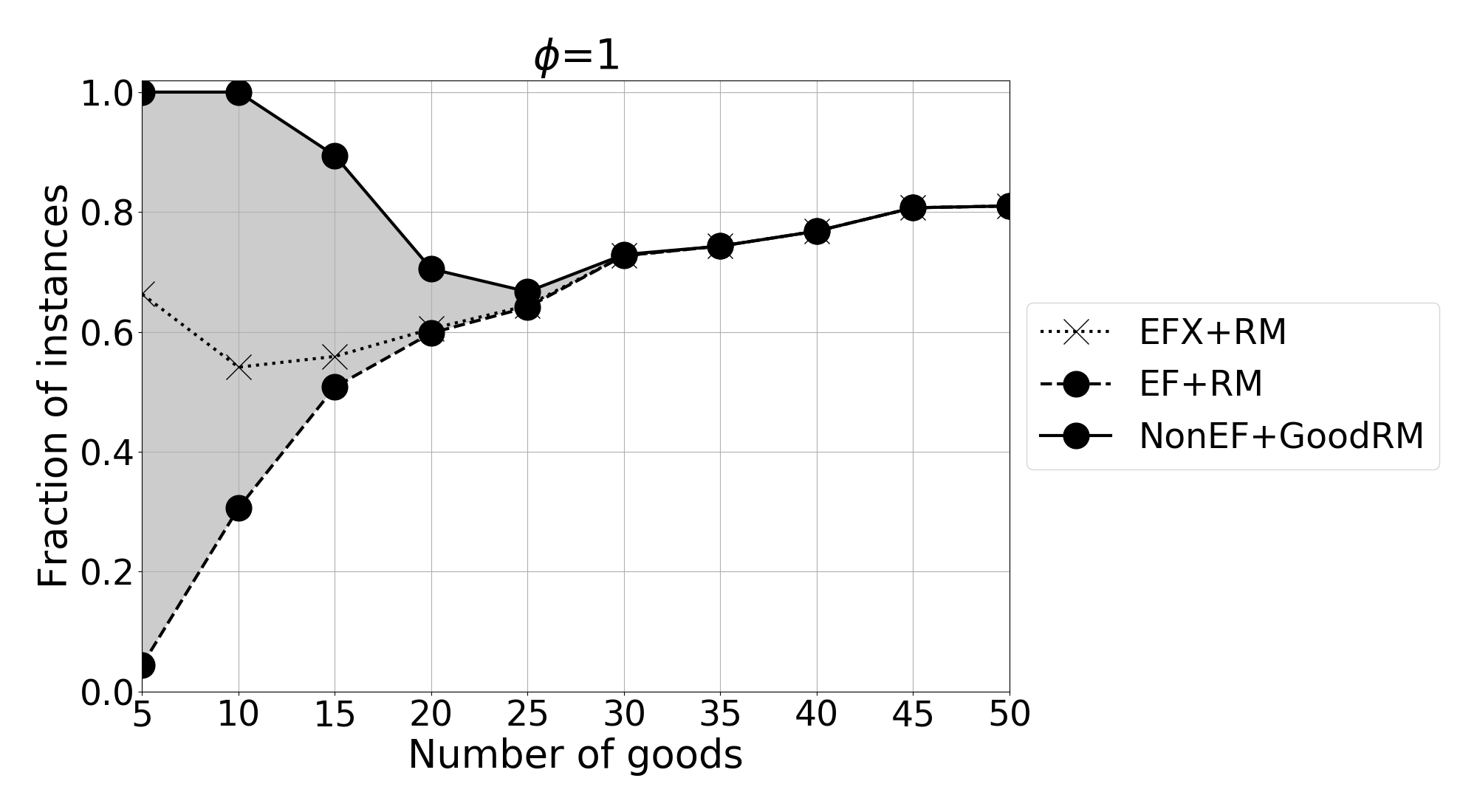}
	\caption{The plot compares the fraction of instances that do not admit an \EF{} allocation and the fraction of good instances that admit an \RM{} allocation where some agent is allocated one good or less. The number of agents is fixed $(n=5)$, and their preferences follow the Mallows model with dispersion parameter $\phi$.}
	\label{fig:mallows_rm_vs_efxrm}
\end{figure}

\Cref{fig:mallows_rm_vs_efxrm} presents our experimental results:
For each value of $\phi$, a solid line represents the fraction of instances that are either ``good'' but do not admit an \EF{} allocation or admit an \EF{}+\RM{} allocation, serving as an upper bound on the fraction of instances that admit an \EFX{}+\RM{} allocation. A dashed line represents the fraction of instances that admit an \EF{}+\RM{} allocation, and a dotted line represents the fraction of instances that admit an \EFX{}+\RM{} allocation.

We observe that for $\phi=1$, as the number of goods increases, the shaded region representing the fraction of ``good'' instances that do not admit an \EF{} allocation shrinks, and is close to zero for larger values of $m$. We believe that this is the reason that as the number of goods increases, the fraction of instances that do not admit an \EF{}+\RM{} allocation but do admit an \EFX{}+\RM{} allocation decreases even as the fraction of instances that admit an \EF{}+\RM{} allocation increases. We observe a similar trend for $\phi=0.25, 0.5$, and $0.75$.

\end{document}